\documentclass[USenglish,oneside]{article}

\usepackage{verbatim}
\usepackage{graphicx,url}
\usepackage[linesnumbered]{algorithm2e}
\usepackage{authblk}
\usepackage{color}
\usepackage{citesort}
\usepackage{wrapfig}
\usepackage{subcaption}
\usepackage{mathtools}
\usepackage{xcolor}
\usepackage{framed}

\usepackage{epsfig}
\usepackage{algorithmic}
\usepackage{amsmath, amsthm,amssymb, amsfonts} 

{\hfil\hskip2em\penalty250\parfillskip=0pt\finalhyphendemerits=0$\qed$%
\par\smallskip\par}

\newtheorem{thm}{Theorem}

\newtheorem{lemma}{Lemma}

\newtheorem{definition}{Definition}

\newcommand{\order}[1]{O(#1)}

\newcommand{\sysname}[1]{DataLair}{\ignorespaces}
\newcommand{\sysoram}[1]{DL-ORAM}

\DeclarePairedDelimiter\floor{\lfloor}{\rfloor}

\begin{document}

\author{Anrin Chakraborti}
\author{Chen Chen}
\author{Radu Sion}
\affil{Stony Brook University}

\title{\huge DataLair: Efficient Block Storage with Plausible Deniability against Multi-Snapshot Adversaries}
\maketitle
\abstract{
Sensitive information is present on our phones, disks,
watches and computers. Its protection is essential. Plausible
deniability of stored data allows individuals to deny that their
device contains a piece of sensitive information. This constitutes
a key tool in the fight against oppressive governments and
censorship.
Unfortunately, existing solutions, such as the now defunct
TrueCrypt [5], can defend only against an adversary that can
access a user’s device at most once (“single-snapshot adversary”).
Recent solutions have traded significant performance overheads
for the ability to handle more powerful adversaries able to
access the device at multiple points in time (“multi-snapshot
adversary”). 
In this paper we show that this sacrifice is not necessary. We
introduce and build DataLair\footnote{A preliminary version of this paper 
was present as a poster~\cite{datalair_poster} with an overview of the technical solution described here, and 
as a full paper \cite{datalair}. This paper additionally
addresses security concerns raised by a recent private email conversation \cite{detworam} in Section \ref{oram_pets_fix}.}, a practical plausible deniability
mechanism. When compared with existing approaches, DataLair 
is two orders of magnitude faster for public data accesses, and 5 times
faster for hidden data accesses.
An important component in DataLair is a new write-only
ORAM construction which improves on the complexity of the state of
the art write-only ORAM by a factor of $\order{logN}$, where N denotes the underlying storage disk size.}


%

 \section{Introduction}

With increasing amounts of sensitive data being stored on portable storage 
devices, disk encryption has
become a necessity.  Although full disk encryption (FDE) tools (such as dm-crypt) 
provide protection against unauthorized adversaries attempting to access sensitive data at rest, it
does not allow the owner to deny possession of sensitive data.  This is a
serious challenge in the presence of oppressive regimes and other powerful
adversaries that may want to coerce the user into revealing encryption keys. 
Unfortunately, this is an all-too-common occurrence as illustrated by numerous
examples~\cite{censor_ex1,censor_ex2}, where sensitive data in possession of human rights
activists have been subject to government scrutiny in oppressive regimes,
thus endangering the witnesses.

Plausible deniability (PD) provides a strong defense against such coercion.  A
system with PD allows its users to deny the existence of
stored sensitive information or the occurrence of a sensitive
transaction\cite{McDonald:1999:SSF:647596.731709}.

An example of a plausibly deniable storage solution is the successful, yet
unfortunately now-defunct TrueCrypt \cite{Truecrypt}.  TrueCrypt divides a
disk into multiple ``password-protected'' volumes and allows some of these
volumes to be ``hidden'' in order to store sensitive data.  Password-derived
encryption keys are used to encrypt each such volume.  Upon coercion, a user
can plausibly deny the existence of a hidden volume by simply providing a
valid password for one of the non-hidden ones, thus showing a plausible use
for the disk without revealing the hidden volume data. TrueCrypt stores 
hidden volumes in the free space of non-hidden (public) volumes.  To mask
their existence, TrueCrypt fills all free space with random data and
encrypts the hidden data with a randomized semantically secure encryption scheme
with output indistinguishable from random.  

However, as pointed out by Czeskis \cite{Czeskis:2008:DED:1496671.1496678},
TrueCrypt is not secure against an adversary that can access the disk at
multiple points in time (e.g., multiple security checks or border
crossings).  In such scenarios, an adversary can save a disk snapshot and compare
subsequent snapshots with it.  Changes to free space occurring
between snapshots will suggest the existence of hidden data.
%

%
A major reason why TrueCrypt fails to provide PD against
an adversary with multiple snapshots is because it does not attempt to hide
{\em access patterns}.  The adversary can point out exact locations on disk
that have changed in between snapshots and notice that the apparently free
portion of the disk (potentially containing hidden data) appears altered.

To defeat a multi-snapshot adversary, we need to eliminate all evidence of
hidden data and its corresponding accesses -- for example by ensuring that all 
modifications on the disk are attributable and indistinguishable from the traces of
public data operations.

This means that modifications to apparently free space should be part of
normal behavior of plausible public operations and the traces of hidden data
accesses should be indistinguishable from the traces of public data accesses. 

One effective way to achieve this is to ``cloak'' hidden accesses within
public data accesses by always performing a public operation for every hidden
operation.  Further, oblivious access mechanisms (ORAM) can be used for
randomizing accesses and making them indistinguishable
\cite{Blass:2014:TRH:2660267.2660313}.  Unfortunately, ORAMs come with very
high overheads and reduce overall throughput by orders of magnitude.

Fortunately, a new insight emerges that enables a significant throughput
increase: accesses to public data do not need to be hidden since the
existence of public data is admitted. In fact, revealing access patterns
to public data reinforces deniability as it shows non-hidden disk use to
a curious adversary.  

Consequently, \sysname~ uses this insight to design a significantly more
efficient way to achieve strong PD: protecting only
operations on the hidden data, while ensuring that they are
indistinguishable from operations on public data (thus allowing the user to
claim that all I/O to the disk is due to public accesses). Further, \sysname~ also 
optimizes the oblivious access mechanism deployed for hidden data. 

In summary, public data is accessed (almost) directly without the need to
use an oblivious access mechanism while hidden data accesses are mitigated
through a new throughput-optimized write-only ORAM which significantly
reduces access complexity when compared to existing work
\cite{Blass:2014:TRH:2660267.2660313,cryptoeprint:2013:694}.  As a result,
\sysname~ is two orders of magnitude faster for public data accesses, and 5 times faster for
hidden data accesses, when compared to existing work.

%

%
%

 \section{Related Work}

Plausible deniability (PD) was first proposed in relation to deniable encryption
\cite{Canetti:1997:DE:646762.706165}.  Deniable encryption uses
cryptographic techniques to allow decrypting the same ciphertext to
different plaintexts.  

\smallskip
\noindent
{\bf Filesystem Level PD.~}
For storage devices, Anderson et al. first explored the
idea of steganographic filesystems and proposed two solutions for hiding data in \cite{Anderson98thesteganographic}. 
The first solution is to use 
a set of cover
files and their linear combinations to reconstruct hidden files.  The
ability to correctly compute the linear combination required to reconstruct
a file was based on the knowledge of a user-defined password.  The second solution was to use 
a hash based scheme for storing files at locations determined by the hash of
the filename.  This requires storing multiple copies of the same
file at different locations to prevent data loss.  Macdonald and Kahn
\cite{Mcdonald99stegfs:a} designed and implemented an optimized steganographic filesystem for Linux,
that is derived from the second solution proposed in
\cite{Anderson98thesteganographic}.  Pang et al.  \cite{1260829} further improved on
the previous constructions by avoiding hash collisions and more efficient
storage. 

The solutions based on steganographic filesystem only defend against
a single-snapshot adversary.  Han et al.
\cite{Han:2010:MSF:1920261.1920309} designed a steganographic filesystem
that allows multiple users to share the same hidden file.  Further, 
runtime relocation of data ensures deniability against an
adversary with multiple snapshots.  However, the solution does not scale
well to practical scenarios as deniability is attributed to joint-ownership
of sensitive data.  Defy \cite{defy} is a log structured file system for
flash devices that offers PD using secure deletion. 
Although, Defy protects against a multi-snapshot adversary, it does so by
storing all filesystem related metadata in the memory, which does not scale
well for memory constrained systems with large external storage devices.

\smallskip
\noindent
{\bf Block device level PD.~}
At device-level, disk encryption tools such as Truecrypt \cite{Truecrypt}
and Rubberhose \cite{Rubberhose} provide deniability but cannot protect
against a multi-snapshot adversary.  Mobiflage \cite{6682886} also provides
PD for mobile devices against a single-snapshot
adversary.  Blass et al.  \cite{Blass:2014:TRH:2660267.2660313} were the
first to deal with deniability against a multi-snapshot adversary at device
level.  The solution in \cite{Blass:2014:TRH:2660267.2660313} deploys 
a write-only ORAM for mapping data from logical volumes to
an underlying storage device and hiding access patterns for hidden data
within reads to non-hidden (public) data.

\section{Model}	
\label{pd}

We focus on storage-centric plausible deniability (PD) as in the results
discussed above, but we note that PD has also been
defined in more general frameworks \cite{plaus_Den}.

\smallskip
\noindent
{\bf Plausible Deniability in real life.~}
It is important to understand however that the mere use of a system with PD
capability may raise suspicion!  This is particularly the case if PD-enabled
systems have high overheads or are outright impractical when accessed for
public data storage.  This is why it is important to design mechanisms that
are practical and do not impede the use of the device, especially for
storing non-sensitive data.  We envision a future where all block device
logic is endowed with a PD mode of operation.

\subsection{Preliminaries}
\noindent
{\bf Adversary.~}
We consider a computationally bounded ``multi-snapshot'' adversary that has
the power to coerce the user into providing {\em a} password.  As in
existing research \cite{Blass:2014:TRH:2660267.2660313}, we also assume that the
device user is not directly observed by the adversary during writes -- a
small amount of volatile memory is being used during reads
and writes and is inaccessible to the adversary that can only see static
device snapshots.

\smallskip
\noindent
{\bf Configuration.~} 
While we note that there may be a number of other ways to achieve PD
, our focus is on a practical solution involving storage devices
with multiple logical volumes, independently encrypted with user
password-derived keys.

To protect her sensitive data from adversaries, a user may write it
encrypted to one of these logical volumes (the ``hidden'' volume).  For
PD, the user may choose to also write non-sensitive data
to a ``public'' volume, encrypted with a key which can be provided to an
adversary as proof of plausible device use. 

\smallskip
\noindent
{\bf Logical and physical blocks.~}  
For a block device hosting multiple logical (hidden or public) volumes,
clients address data within each volume using logical block addresses.  The data
is stored on the underlying device using physical block addresses.

\smallskip
\noindent
{\bf Access patterns.~}
We define an {\em access pattern} informally as an {\em ordered sequence of
logical block reads and writes}.  


\smallskip
\noindent
{\bf Write traces.~}
We define a {\em write trace} as the actual modifications to physical blocks
due to the execution of a corresponding {\em access pattern}.

\smallskip
\noindent
{\bf Solution Space.~} 
While there might be multiple ways to achieve PD in a
multi-volume multi-snapshot adversary setting, one idea
\cite{Blass:2014:TRH:2660267.2660313} is to ``hide'' operations to the
hidden volume ``within'' operations to the public volume.  This prevents the
adversary from gathering any information regarding user {\em access
patterns} to the hidden volume (hidden data access) by ensuring that the user can
plausibly attribute any and all {\em writes traces} as accesses to the
public volume (public data access) instead.  On coercion, the user can
provide the credentials for the public volume and thus plausibly deny the
existence of the hidden data.  Arguably, otherwise, in the absence of public
volume operations, an adversary could question the reason for any observed
changes to the space allocated for the hidden volume and then rightfully
suspect the existence of a hidden volume.

\smallskip
\noindent
{\bf Atomicity.~}
As in existing work, a very small number of physical block I/O ops
(corresponding to one \sysname~ read/write operation) are assumed to be
performed as atomic transactions.  The adversary may gain access only after
entire transactions have completed (or rolled back).  This is reasonable
since the user is unlikely to perform any sensitive operation in the
presence of an adversary.

\smallskip
\noindent
{\bf Access pattern indistinguishability.~} 
Computational indistinguishability of traces has been widely discussed in
ORAM literature \cite{pathoram,phantom}.  Further, most ORAMs employ
randomization techniques which ensure that the {\em write traces} generated
due to accesses are indistinguishable from random.  As will be detailed
later, indistinguishability of access patterns is one of the main requirements to achieve 
PD.

We also define the link between access patterns and write traces. 
\begin{definition}
Given two access patterns $\mathcal{O}_0 = \{a_{1}, a_{2}, \ldots , a_{i}\}$
and $\mathcal{O}_1 = \{b_{1}, b_{2}, \ldots , b_{i}\}$ and their
corresponding {\em write traces} $\mathcal{W}_0 = \{w_{1}, w_{2}, \ldots ,
w_{i}\}$ and $\mathcal{W}_1 = \{y_{1}, y_{2}, \ldots , y_{i}\}$,
$\mathcal{O}_0$ is called indistinguishable from $\mathcal{O}_1$ iff. 
$\mathcal{W}_0$ is computationally indistinguishable from $\mathcal{W}_1$.
\end{definition}

\subsection{Defining Plausible Deniability (PD-CPA)}

We model PD as a ``chosen pattern'' security game, PD-CPA \footnote{The
similarity with IND-CPA is intentional.  The access patterns correspond to the
``plaintexts'' in an access pattern privacy setting.} for the block device. 
Since we focus on mechanisms that ``hide'' operations to the hidden volume
``within'' operations to the public volume, we also define an
implementation-specific parameter establishing the number of operations that can be
hidden within a public operation.  Specifically, $\phi$ is the ratio of the
number of hidden operations that can be performed with a public operation
such that the {\em write traces} due to the public operation with a hidden
operation is indistinguishable from the {\em write traces} for the same
operation without the hidden operations. This paper proposes a solution with 
$\phi=1$ to ensure that each public operation performs {\em one} hidden operation.

We define PD-CPA($\lambda$, $\phi$) , with security parameter $\lambda$
between a challenger and an adversary as follows:

\begin{framed}
\begin{footnotesize}

\begin{enumerate}

\item Adversary $\mathcal{A}$ provides a storage device $\mathcal{D}$
({\em the adversary can decide its state fully}) to challenger $\mathcal{C}$.

\item $\mathcal{C}$ chooses two encryption keys $K_{pub}$ and $K_{hid}$ using
security parameter $\lambda$ and creates two logical volumes, $V_{pub}$ and
$V_{hid}$, both stored in $\mathcal{D}$.  Writes to $V_{pub}$ and $V_{hid}$
are encrypted with keys $K_{pub}$ and $K_{hid}$ respectively. $\mathcal{C}$ 
also fairly selects a random bit $b$. 

\item $\mathcal{C}$ returns $K_{pub}$ to $\mathcal{A}$.

\item The adversary $\mathcal{A}$ and the challenger $\mathcal{C}$ then
engage in a polynomial number of rounds in which:

\begin{enumerate} 

\item  $\mathcal{A}$ selects two access patterns $\mathcal{O}_0$ and
$\mathcal{O}_1$ with the following restriction: 

\begin{itemize}
\item $\mathcal{O}_1$ and $\mathcal{O}_0$ include the same writes to $V_{pub}$
\item Both $\mathcal{O}_1$ and $\mathcal{O}_0$ {\em may} include writes to $V_{hid}$
\item $\mathcal{O}_0$ or $\mathcal{O}_1$ {\em should not} include more writes to $V_{hid}$ than $\phi$ times the number of operations to $V_{pub}$ in that sequence.
\end{itemize} 

\item $\mathcal{C}$ executes $\mathcal{O}_b$ on $\mathcal{D}$ and sends a
snapshot of the device to $\mathcal{A}$.

\item $\mathcal{A}$ outputs $b^{'}$.

\item $\mathcal{A}$ is said to have ``won'' the round iff. $b^{'} = b$.

\end{enumerate}


\end{enumerate}
\end{footnotesize}
\end{framed}

We note that the restrictions imposed on the adversary-generated access
patterns are necessary to eliminate trivially-identifiable cases.   E.g.,
allowing different writes to $V\_pub$ in $\mathcal{O}_0$ and $\mathcal{O}_1$ 
would allow the adversary to distinguish between the two sequences by simply
decrypting the contents of $V\_pub$ (using the known $K_{pub}$) and comparing the decrypted {\em data}. 
Fortunately, a similar comparison is not possible for $V_{hid}$ since 
$K_{hid}$ is not accessible to the adversary.

The restriction imposed by $\phi$
ensures that the adversary may not trivially distinguish between
$\mathcal{O}_0$ and $\mathcal{O}_1$ by providing sequences of different ``true 
lengths'' -- the number of actual blocks modified 
in the write trace generated by a given access pattern. Specifically, for a given $\phi$, PD-CPA assumes 
that the true length of a sequence with $k$ writes to $V\_pub$ is $k \times \phi$. Since,  $\mathcal{O}_0$ and $\mathcal{O}_1$ 
have the {\em same} writes to $V\_pub$, their true lengths are the same. 
The additional $\phi$ writes in the corresponding write 
traces generated for the sequences allows hiding writes to $V\_hid$. Thus, 
if the number of writes to $V\_hid$ exceed the number allowable by $\phi$, the sequences become 
trivially distinguishable by their true lengths.

%

\smallskip
\noindent
{\bf Relationship with existing work.~} 
PD-CPA is similar to the hidden volume encryption game in
\cite{Blass:2014:TRH:2660267.2660313} with the notable difference that
PD-CPA empowers the adversary further by giving her the ability to choose
the input device.  This consideration is regarding a practical scenario
where an oppressive regime officer might be aware of particular underlying
properties of a storage device.  Thus, a PD-CPA secure solution should not be reliant
on the properties of a particular kind of storage device that can be used
by the challenger, as that in itself would be suspicious to the
adversary\footnote{For example, \cite{hiding_flash} exploits the variations in the
programming time of a flash device to hide data by encoding bits in the
programming time of individual cells.}.

\begin{definition}
A storage mechanism is ``plausibly deniable'' if it ensures that
$\mathcal{A}$ cannot win any round of the corresponding PD-CPA game with a
non-negligible advantage (in the security parameter $\lambda$) over random
guessing.
\end{definition}

\subsection{Necessary and sufficient conditions for PD-CPA}

\sysname~ provides a plausibly deniable solution defeating a PD-CPA
adversary.  We note informally here (and prove later) that the
following conditions are necessary and sufficient to ensure PD-CPA security.

\begin{enumerate}
 \item Indistinguishability between hidden data write access patterns (``access pattern indistinguishability'', HWA).
 \item Indistinguishability between write traces that include public data accesses with one (or more) hidden data accesses, and 
      the {\em same} public data accesses without {\em any} hidden data accesses (``access type indistinguishability'', PAT). 

\end{enumerate}

\smallskip
\noindent
{\bf Indistinguishability between hidden write access patterns (HWA).~}
Recall that in PD-CPA, the adversary can include writes to $V_{hid}$ in both
$\mathcal{O}_0$ and $\mathcal{O}_1$.  If the write traces were
distinguishable, the two sequences would become distinguishable to an
adversary providing different accesses to $V_{hid}$ in the sequences.

Note that HWA indistinguishability also ensures that writes to the same
logical location in $V_{hid}$ in two different rounds of PD-CPA results in
write traces that are independently distributed and thus prevents an
adversary from ``linking'' accesses to the same logical location in
successive rounds.  In the absence of this ``unlinkability'', an adversary
could provide the same accesses to $V_{hid}$ as part of $\mathcal{O}_1$ in
successive rounds with only writes (possibly different) to $V_{pub}$ as
$\mathcal{O}_0$.  On observing the same write traces for successive rounds,
the adversary could correctly predict that the sequence executed is
$\mathcal{O}_1$.

\smallskip
\noindent
{\bf Indistinguishability between public access {\em write traces} (PAT).~}
In PD-CPA, an adversary can provide $\mathcal{O}_0$ with only accesses to
$V_{pub}$ and $\mathcal{O}_1$ with accesses to both $V_{pub}$ and $V_{hid}$
and win trivially if the sequences' write traces were distinguishable.  In
effect, to ensure PD-CPA security, {\bf in case $\mathcal{O}_0$ and
$\mathcal{O}_1$ include the {\em same} public data operations, they should
be indistinguishable}.  Note that both sequences should contain the {\em
same} public data operations since otherwise they are trivially
distinguishable on the basis of any additional public operation that is
performed, as discussed above.  This is why the write trace due to a sequence of public data
operations plus (one or more) hidden data write(s) should be
indistinguishable from the write trace due to accesses including the {\em
same} public data operations without any hidden data write(s).

\begin{thm}
\label{model:thm}
A storage mechanism is PD-CPA secure iff.  it guarantees
indistinguishability between hidden write access patterns (HWA) 
{\em and} indistinguishability between public operations with and
without hidden writes (PAT).
\end{thm}

\begin{proof}
HWA straightforwardly ensures that even if the same logical locations in
$V_{hid}$ are written by two access patterns in two rounds, their write
traces are independent.  Otherwise, an adversary could win one of the games
by observing the same write traces for the same writes to $V_{hid}$ in
subsequent rounds.

Also, in absence of PAT, an adversary could win PD-CPA by providing
$\mathcal{O}_0$ and $\mathcal{O}_1$ with the {\em same} public data
operations but with and without writes to $V_{hid}$ respectively.  Thus, PAT
is a necessary condition to ensure PD-CPA security.

We now show that HWA and PAT are also sufficient for PD-CPA.
First, note that HWA ensures that writes to locations in hidden volume $V_{hid}$ map to
physical locations selected independently of their corresponding logical
locations.  Logical and physical locations are uncorrelated.  An adversary
cannot determine the logical locations corresponding to observed modified
physical locations.

Second, observe that in the context of the PD-CPA game, PAT ensures that traces due to
combined writes to $V_{hid}$ and $V_{pub}$ can be attributed to writes
corresponding to $V_{pub}$ only.  And, as the writes to $V_{pub}$ are the
same in both the sequences, the adversary cannot distinguish between the
write traces non-negligibly better than guessing.

Now, consider a PD solution, $S$ which provides both HWA
and PAT.  Also, consider an adversary $\mathcal{A}$ that wins PD-CPA against
$S$ selecting two sequences $\mathcal{O}_0$ and $\mathcal{O}_1$.  Since,
$\mathcal{O}_0$ and $\mathcal{O}_1$ differ only in the writes to $V_{hid}$
(writes to $V_{pub}$ are the same), either of the following holds:

\begin{itemize}

\item $\mathcal{O}_0$ and $\mathcal{O}_1$ contain different writes to $V_{hid}$ and they are distinguishable from the corresponding 
write traces in direct contradiction to the HWA property of $S$.

\item $\mathcal{O}_0$ contains writes to $V_{hid}$ and $\mathcal{O}_1$ does
not contain writes to $V_{hid}$, and the corresponding write traces are
distinguishable. This implies that traces due to combined writes to
$V_{hid}$ and $V_{pub}$ in $\mathcal{O}_0$ are distinguishable from traces
with only the {\em same} writes to $V_{pub}$ in $\mathcal{O}_1$, in direct
contradiction to the PAT property of $S$.

\end{itemize}

Note that ensuring HWA without PAT and vice versa is not sufficient for 
PD-CPA since either of the above two cases will allow the
adversary to win with non-negligible advantage.  

\end{proof}

\section{Access Pattern Indistinguishability} 
\label{ORAM}

Section~\ref{pd} shows that one of the necessary conditions to plausibly deny the existence  of a logical volume to a multi-snapshot adversary, 
is to ensure indistinguishability of hidden data write access patterns (HWA). 

A straightforward solution here is to use an oblivious RAM (ORAM) which
allows a client/CPU to hide its data access patterns from an untrusted
server/RAM hosting the accessed data.  Informally, ORAMs prevent an
adversary from distinguishing between equal length sequences of queries made
by a client to the server.  This usually also includes indistinguishability
between reads and writes.  We refer to the vast amount of existing
literature on ORAMs for more formal definitions
\cite{Stefanov:2013:POE:2508859.2516660,Goldreich:1996:SPS:233551.233553}.
\smallskip
\noindent
{\bf Write-ony ORAM.~}
As noted by Blass {\em et al.} \cite{Blass:2014:TRH:2660267.2660313}, an
off-the-shelf full ORAM is simply too powerful since it protects both read
and write accesses -- while for PD
only write access patterns are of concern.  In
this case, a write-only ORAM~\cite{cryptoeprint:2013:694,Blass:2014:TRH:2660267.2660313} 
can be deployed which
provides access pattern privacy against an
adversary that can monitor only writes.

\smallskip
\noindent
{\bf Logical vs. practical complexity.~}
ORAM literature traditionally defines access complexity as the number 
of logical blocks of data accessed per I/O. This allows optimizations 
by using logical blocks of smaller size~\cite{Shi_obliviousram,pathoram}.
However, it is important to note that standard off-the-shelf block-based storage devices can access data only in units of 
physical blocks (sectors). For example, accessing a 256 byte logical 
block still entails accessing the corresponding 4KB block from the disk.
Thus, in the  context of locally deployed  block-based storage devices, practical complexity needs to be measured
in terms of the number of physical blocks (sectors) 
accessed per I/O.

\smallskip
\noindent
{\bf HIVE-ORAM.~}
The most bandwidth-efficient write-only ORAM is the construction by Blass
{\em et al.} \cite{Blass:2014:TRH:2660267.2660313} (further referred to as
HIVE-ORAM).  HIVE-ORAM \cite{Blass:2014:TRH:2660267.2660313} maps data from
a logical address space uniformly randomly to the physical blocks on the underlying
device.  The position map for the ORAM is recursively stored in
$\order{logN}$ smaller ORAMs, a standard technique introduced in
\cite{Shi_obliviousram}. The recursive technique reduces the logical block access complexity for the position map 
by storing the position map blocks in logical blocks of {\em smaller} sizes.
Under this assumption, HIVE-ORAM \cite{Blass:2014:TRH:2660267.2660313} accesses a {\em constant} number
of logical blocks per ORAM operation at the cost of some overflow that is
stored in an in-memory stash. 

Unfortunately, as noted in~\cite{Blass:2014:TRH:2660267.2660313}, 
with physical blocks of uniform size, HIVE-ORAM has a practical block read complexity (number of
blocks read) of $\order{log_{\beta}N}$ and a practical block write complexity (number
of blocks written) of $\order{log_{\beta}^{2}N}$ where $\beta= B/2*addr$, $B$
is the physical block size in bytes and $addr$ is the size of one logical/physical
block address in bytes. This is because to perform a write, HIVE-ORAM~\cite{Blass:2014:TRH:2660267.2660313} 
needs to update all $log_{\beta}N$ position map ORAMs recursively and updating each ORAM 
requires $\order{log_{\beta}N}$ accesses to find free blocks. More specifically, 
to find free blocks in an ORAM, a constant number of randomly chosen
 physical blocks are selected and then the position map for that ORAM is checked to determine free blocks 
 within the sample. Consequently, with $\order{log_{\beta}N}$ block read complexity of each position map ORAM, 
 the overall write complexity of HIVE-ORAM is $\order{log_{\beta}^{2}N}$.

\smallskip
\noindent 
{\bf DL-ORAM.~}
We propose \sysoram~, a new efficient write-only ORAM scheme with a practical block write
complexity of $\order{log_{\beta}N}$. Similar to \cite{Blass:2014:TRH:2660267.2660313}, \sysoram~ maps 
blocks from the logical address space to uniformly random blocks in the 
physical address space. The \sysoram~ construction however is significantly different and 
incorporates two key optimizations. 

First, \sysoram~ eliminates the need for recursive position map ORAMs by  
storing the position map as a B+ tree, indexed by the logical block addresses.  
The tree is stored along with the data within the same ORAM address space, thus judiciously 
utilizing space.  
Second, \sysoram~ uses a novel $\order{1}$ scheme for identifying uniformly random free blocks using
an auxiliary encrypted data structure.  This allows writes with
$\order{log_{\beta}N}$ communication complexity.  We detail below.


\subsection{Position Map}
In \sysoram~, all blocks are stored encrypted with a semantically secure cipher. Further, \sysoram~ stores the logical to
physical block address mappings in a B+ tree indexed on the logical
data block address. Logical and physical addresses are 
assumed to be of the same size. The position map tree is stored on the
device with {\em each} node stored in an {\em individual} physical block. Each 
{\em leaf} node stores a 
sorted list of logical block addresses along with the corresponding 
physical block addresses they are mapped to. 

Further, the leaf nodes are ordered from left to right, e.g., the left-most leaf node
contains the entries for logical block addresses in the range of 1 to $\beta$.
This ensures that the $(i/\beta)$-th leaf node always contains the mapping for logical block address 
$i$. If a logical block is currently not mapped (e.g. 
 when the block hasn't been written to yet by the upper
layer), the entry corresponding to that address in the map is set to null. 

For traversal, each 
internal node stores only the list of physical block address of the node’s
children. Note that since the leaves are ordered on logical addresses, the index within an internal node 
determines the child node to visit next in order to retrieve a path corresponding to a particular 
logical block address. Searching for a particular logical block 
 address mapping requires reading all the
blocks along a path from the root to the leaf that contains the
mapping for that logical block ID. As each node is stored in a physical block,  
the number of block addresses that can be 
stored in each leaf node is bound by the physical block size, $\beta$. 
Consequently, the depth of the tree is bounded by
$log_{\beta}(N)$ with a fanout of $\beta$.  Thus, querying the map for a 
particular logical to physical block address mapping requires $log_{\beta}(N)$ 
block accesses.

\begin{figure}
\centering
 \includegraphics[scale=0.18]{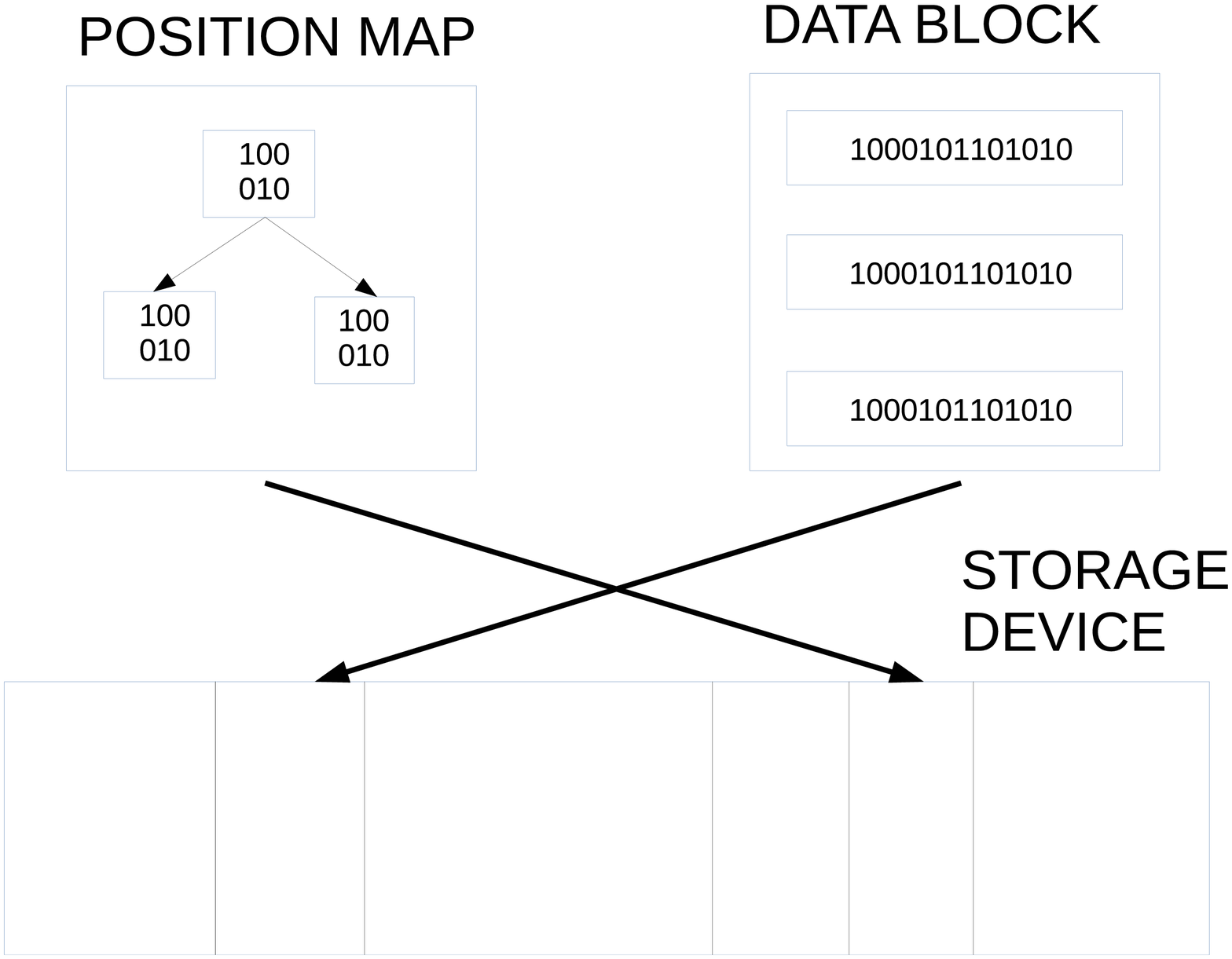}
 \vspace{-0.3cm}
 \caption{DL-ORAM. The position map blocks are randomly interleaved with the data blocks within the same address space.
 \label{DL_ORAM}}
 
\end{figure}

%
The position map shares the same physical address space with 
the ORAM data blocks. Specifically, the B+ tree blocks are assigned a logical block address 
and are written to random physical blocks interleaved with the data blocks,
using the ORAM protocols. 
The physical location of the tree root is stored at a known location outside 
the ORAM address space (intuition explained later). Semantic security of the cipher used ensures
that the position map blocks are indistinguishable from the data blocks in
the ORAM. 


\sysoram~  supports two operations: $read\_oram$ and $write\_oram$. Detailed 
pseudocode can be found in the Appendix.

\smallskip
\noindent
$\mathbf{read\_oram(id)}$ returns the data in the block with logical block address $id$. It 
locates the mapping for $id$ in the B+ tree and returns the data stored in the corresponding physical block. 

\smallskip
\noindent
$\mathbf{write\_oram(id,d)}$ writes data $d$ to the block with logical block address 
$id$. It first determines the entry corresponding to $id$ in the 
B+ tree and then writes data $d$ to 
a new free block. Finally, the map is updated corresponding to $id$. 
The mechanism for finding a free block is described next in 
Section~\ref{ORAM:findingfree}

Finally, updating the map requires recursively writing all the nodes along the 
the specific path of the B+ tree to new free blocks to ensure indistinguishability between 
map blocks and data blocks. Specifically, the updated leaf node (after adding the entry for $id$) is 
written to a selected new free block. This results in
its parent node being modified and being mapped to a new free location. 
Consequently, on an ORAM write, all the blocks from the root of
the map to the corresponding leaf node are modified and remapped to new
locations. 

To ensure that the recursion terminates, 
the physical location of the tree root is always stored at a known fixed location as described before. 
For \sysoram~ , this information is held at a fixed
location on the disk where the physical block address of the root is written encrypted and 
is modified for each access. Once the root data is also modified and remapped to a new 
physical block, this information is updated. 

\subsection{Finding Free Blocks}
\label{ORAM:findingfree}
%
A major challenge for write-only ORAMs is selecting 
uniformly randomly free block from the distribution
of all free blocks on the device for writing data. Writing thusly eliminates
correlations between the logical data block addresses and their
physical locations, thus rendering an adversary incapable of
linking modifications to physical blocks with the corresponding 
modifications to logical blocks. Even when data already in
the ORAM needs to be updated, it is relocated to a new
random location to prevent correlation with previous writes.


\sysoram~ deploys a new $\order{1}$ scheme for identifying
uniformly random free blocks using an auxiliary encrypted data structure. 
This allows writes with $\order{log_{\beta}(N)}$ access complexity. 
This is achieved by storing free block 
addresses in a novel encrypted data structure -- the {\em free block matrix} (FBM). The FBM 
is designed with the two following properties -- i) it allows retrieval of uniformly random 
elements in $\order{1}$ accesses, and ii) it does not reveal the actual number 
of elements that it contains at any given point in time. Ensuring (i) allows efficient 
retrieval of random free block addresses in our scheme (as we detail later). (ii) ensures 
that the actual number of data blocks in the ORAM is not revealed thus preventing
possible correlations.

\smallskip
\noindent
{\bf Free Block Matrix (FBM).~}
The FBM is an encrypted $\beta \times N/\beta$  matrix (where $N$ is the total number of
physical blocks in the ORAM and $\beta$ is number of physical block addresses that can be 
written to one disk block) that stores the addresses of all currently free blocks. 
The columns of the matrix are stored in {\em consecutive}
disk blocks outside the ORAM, each containing $\beta$ free block addresses (Figure
\ref{fbm_design}). The coordinates in the matrix where each block address is 
stored is randomized independent of the address. 
Since all blocks are free at initialization, the FBM is
initialized with all $N$ entries containing a randomly chosen disk
block address.

\begin{figure}
\centering
\includegraphics[scale=0.18]{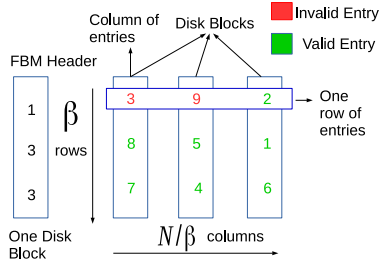}
\caption{FBM design. The FBM is a matrix with $\beta = B/addr$ rows and $N/\beta$ columns, containing $N$ entries. 
 Each entry of the matrix is a 
 physical block ID. Each column of the FBM is stored in a disk block. 
 A special ``FBM header'' array is in  stores the number of physical block IDs per row. The FBM requires $N/\beta$  + 1 disk blocks for storage -- 
 $N/\beta$ columns and 1 FBM Header. The figure illustrates an example FBM configuration with 9 entries. The first row has two invalid entry as indicated 
 by the FBM header. \label{fbm_design}}
\end{figure}

As physical blocks are used for writing data, the corresponding block addresses need 
to be removed from the FBM. This is achieved by {\em invalidating} entries (described next)
when the corresponding blocks are used for ORAM writes. In this case, 
to track the number of valid entries in the FBM, the {\bf FBM header} block 
stores a block-sized array tracking the number of valid entries per FBM {\em row}. 

More specifically, the FBM header contains an entry for each row indicating the number 
of valid entries currently present in that row. Since each row can contain $N/\beta$ valid entries (corresponding 
to the number of columns), each entry in the FBM header is of 
size $log(N/\beta)$. For $\beta$ rows, the total number of entries in the FBM header is 
$\beta \times log(N/\beta)$. Also, $\beta <= B/logN$ with $B$ as the block size and 
each physical address of length at least equal to $logN$. Then for $\beta > 1$, 
\begin{center}
$\beta \times log(N/\beta) \leq \beta \times logN \leq  B/logN \times logN \leq B$
\end{center}

Hence, the FBM header always fits in one physical block if $B > logN$. 
With standard 4KB block size on off-the-shelf storage disks,
the FBM header requires more than one block ($B < logN$) only if the disk has more than $2^{4096}$ blocks in total

Figure~\ref{fbm_design} describes the FBM design.

\begin{figure}
\centering
\includegraphics[scale=0.20]{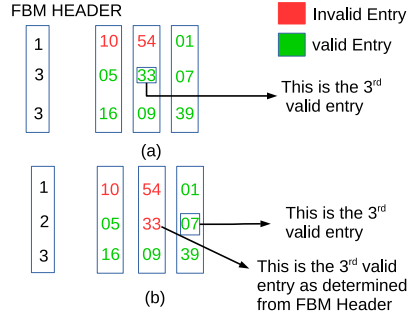}
\caption{Example of FBM with 9 entries. In (a), the FBM header correctly determines that the $3^{rd}$ valid entry is at (2,2) since there is 1 
valid entry in the first row and 3 in the second. In (b), the $3^{rd}$ valid entry is determined incorrectly since the FBM header correctly 
indicates that the second row has 2 valid entries but does 
not show which columns have the valid entries.\label{fbm_prob}}
\end{figure}




\smallskip
\noindent
{\bf Selection from the FBM.~}
The FBM allows selection of uniformly random block addresses with 
two disk block accesses as follows --
select a random $i$ in the range of the total number of valid entries 
in the FBM. Then, determine the coordinate of the $i^{th}$ entry in the FBM
by counting valid entries {\em row-wise} -- this is straightforward since 
the header stores the number of valid entries per row. The block address 
at that coordinate is then retrieved from the corresponding location.

If the block is subsequently used (possibly for a new write), it is {\em invalidated} by reducing the valid entry
count for the corresponding row in the FBM header. Consequently, for subsequent accesses the FBM header indicates that 
the row has one less valid entry. Note that invalidating an entry does not entail removing it 
from the disk block storing the FBM column, rather the FBM header is updated to ensure 
that the entry is not used in subsequent accesses. 

An important condition for correctness of this mechanism is to  ensure that 
a location determined from the FBM header always contains a valid entry.
A possible scenario where this might be violated is shown in Figure~\ref{fbm_prob} where 
the FBM incorrectly points to a location that contains an invalid entry. The problem 
in this case is due the presence of invalid entries between valid entries (as in Figure~\ref{fbm_prob}(b)). 
Here, the FBM header correctly indicates that the second row has 2 valid entries but does 
not indicate the columns with the valid entries. Determining the exact coordinate at which the 
valid entry is present is non-trivial given that the header only records the {\em number} of valid entries in a row but {\em not} 
their corresponding locations. Consequently, with {\em gaps} between valid entries in the FBM (as in Figure~\ref{fbm_prob}), 
the selection mechanism described above can erroneously return block addresses (corresponding to invalid entries) 
that are already being used.
 
 \begin{figure}
 \centering
  \includegraphics[scale=0.18]{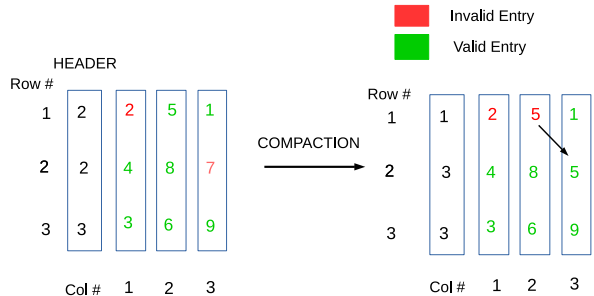}
  \caption{Compacting the FBM. Invalid entries appear before valid entries when entries are traversed {\em row-wise}. 
  In the example, 7 was selected randomly from (2,3) which was replaced by 5 from (1,2) to ensure compactness.\label{compaction_example}} \end{figure}

\smallskip
\noindent
{\bf Compacting the FBM.~}
To ensure correctness, \sysoram~ maintains the following invariant --
all {\bf invalid} entries in the FBM
appear {\bf before} all {\bf valid} entries when 
the FBM is traversed {\em row-wise}.
 This requires all valid entries in a particular row to be invalidated 
first before invalidating an entry from the next row in sequence -- e.g., all entries in the first row are invalidated
before the second row etc. 
Since, entries are selected randomly from the FBM, \sysoram~ performs an extra {\em compaction} step 
to maintain the invariant.

The {\em compaction} replaces a randomly selected valid entry from the FBM 
by another valid entry which is selected {\em row-wise} from the FBM.
This replacement entry is the {\em first} valid entry that is encountered while 
counting valid entries per row sequentially 
from the FBM header. After the replacement, the FBM header is updated.

For example, in Figure~\ref{compaction_example}, 7 is uniformly randomly selected (and invalidated) 
from coordinate (2,3) in the FBM using the protocol described above. This leads to a configuration that
violates the invariant. Then for compaction, 5 is copied from coordinate (1,2) to (2,3) and the entry at (1,2) is invalidated. 
Note that the entry at (1,2) is the {\em first} valid entry encountered while traversing the FBM row-wise. 
Also, even though the FBM now contains duplicate block addresses at (1,2) and (2,3), the block address 
at (1,2) will not be used subsequently during uniform selection since the FBM header correctly indicates 
that row 1 has one valid entry. Due the invariant, this single valid entry has to be present at (1,3).

Compaction straightforwardly ensures the invariant since randomly selected entries are replaced by 
entries that are selected {\em row-wise} and subsequently invalidated. 
To summarize, the compaction modifies two blocks -- the header and the block from which the random entry is 
selected and replaced. 

Moreover, in certain cases this compaction is not required. For instance, when a free block address 
is to be added to the FBM and a new random free block address is required in turn. Recall that this is 
indeed the case for updates to \sysoram~ blocks -- the updated data is written to a new free block while 
the block containing the previous data is now free and the corresponding free block address can 
be added back to the FBM. In this case, the new randomly chosen free block address can be directly
replaced by the free block address that is to be added back to the FBM. Note that due to the direct replacement, 
the header does not need to be updated. 

In order to ensure that 
such an access looks indistinguishable from an access with compaction -- otherwise an adversary could 
identify accesses which are simply updates to existing blocks and thus deduce the actual number of data blocks 
in the ORAM -- the FBM header is reencrypted along with the modification to the block where the replacement takes place. 
Due to semantic security, reencrytion of the header is indistinguishable from actual modification during compaction. 

Finally, note that entries are added back to the FBM only as part of 
updates (as described above). \sysoram~ does not support deletes 
since modern filesystems do not indicate deletes to block devices. Specifically, FS such as ext4 only 
update its internal metadata to track blocks that have been deleted while the deleted data is 
overwritten only when the block is subsequently allocated for writing new data. Thus, when \sysoram~ 
is used with an overlying filesystem mounted logical volume (the deployment scenario considered here),
deletes are logically equivalent to updates and the deleted data is updated with new data in future accesses.


\smallskip
\noindent
{\bf Uniform free block selection}
\label{oram:free_block}
As discussed, the FBM provides an efficient mechanism for determining uniformly random 
free block addresses for ORAM writes with $\order{1}$ access
complexity. 
Unfortunately, writing to blocks only selected from the FBM for every write, can result in 
certain blocks in the ORAM never being modified. 
For instance, consider a block containing data that is never updated. 
This block is never subsequently modified in the ORAM after the first write since 
the block never becomes free again. 
This can leak subtle correlations since an adversary can differentiate between disk blocks 
that are being updated frequently from locations that are not updated since they 
were first written. 

To solve this, \sysoram~ deploys an $\order{1}$ free block selection scheme using the FBM, that 
modifies uniformly random locations in the ORAM and thus prevents any block level correlations. 
The intuition is to always modify $k$ (a chosen constant) random disk blocks while 
also finding a free block. 

For every ORAM write, \sysoram~  creates two sets, each with $k$ items as follows -- 

\begin{enumerate}
\item A {\em free set} that contains $k$ randomly chosen free block addresses from the FBM.
\item A {\em random set} that contains $k$ randomly 
chosen block addresses out of all the blocks in the ORAM.  
\end{enumerate}

In case of duplicate items between the sets, one of the duplicates
is randomly discarded while a new item is selected either to the free set 
or the combined set depending on the set from which the item was discarded.
Then, items from the two sets are merged randomly to create a {\em combined set}. 

\smallskip
\noindent
{\bf Selection Protocol.~}
Using the combined set \sysoram~ , executes the following protocol --

\begin{enumerate}
\item Select an item randomly from the combined set.
\item If the item, $i$ selected in step 1 originally belonged to the {\em free set}, 
use the corresponding block for the
 ORAM write. Otherwise, reencrypt the block corresponding to the 
block address selected. Remove the item from the combined set.
\item Repeat steps 1-3, $k$ times.
\item For all the remaining items, $i$, that are also in the {\em free set}, replace the addresses back in the FBM from where they were selected. 
\end{enumerate}

\begin{lemma}
\label{lemma_rand_ind}
The sample of $k$ physical blocks modified for every ORAM write
(due to the selection protocol) is indistinguishable from a sample of $k$ blocks 
selected uniformly at random out of the $N$ physical blocks in the ORAM.
\end{lemma}

\begin{proof}
The idea here is to show that the probability of a block $x$ being selected (and modified) randomly in a sample of $k$ 
blocks out of $N$, is the same as the probability of $x$ being one of the $k$ blocks chosen by the selection 
protocol. Let $X$ denote the event of $x$ being chosen uniformly at random in a sample of $k$ blocks out of $N$. Then, 
straightforwardly $Pr[X] = k/N$. 

Next, consider the process of building the {\em free set}.
Let $Pr[x \in FBM]$ denote the probability that $x$ is currently in the FBM. Since the 
FBM is initialized with all $N$ block addresses and for each access, entries are selected (and invalidated) uniformly 
randomly, all $N$ block addresses are equally likely to be present and valid in the FBM 
during the current access. More specifically, all the $N$ initial entries are equally likely to have been invalidated 
in the writes that have preceded this access. Thus, $Pr[x \in FBM] = f/N$ when $f$ is the current number of valid entries 
in the FBM. 

Let $E_1$ be the event that $x$ is selected into the {\em free set} -- $x$ is 
one of $k$ randomly chosen block addresses selected from the FBM to form the free set.
This conditionally depends on $x$ being present in the FBM. Thus, $Pr[E_1] = Pr[x \in k | x \in FBM] = Pr[x \in k ] \times Pr[x \in FBM] = k/f \times f/N = k/N$ 
since the event that $x$ is selected as one of the $k$ items from the FBM is independent 
of $x$ being present in the FBM.

Now, since the {\em random set} is created by selecting $k$ random block address out of all $N$, 
the event that $x$ is selected into the random set, $E_2$ straightforwardly has the probability, 
$Pr[E_2] = k/N$. Thus, $x$ has equal probability of being added to the {\em combined set} from 
either the {\em free set} or the {\em random set}. Note that the combined set has size $2k$.

Let $X^{'}$ be the event of $x$ being selected for modification by the protocol. The goal here 
is to show $Pr[X] = Pr[X^{'}]$. 
$X^{'}$ depends on $x$ being chosen either to free set or the random set ($E_1$ and $E_2$) {\em and} 
$x$ being selected in one of the rounds of the protocol, denoted by event $Y$. Thus, $Pr[X'] = Pr[Y] \times Pr[E_1 U E_2]$.

Let $Y = Y_1 + Y_2 + \ldots + Y_k$ where $Y_i$ is the event that $x$ is selected in the 
$i^{th}$ round of the protocol. Also, since $x$ can be selected only once (selection without replacement), probability 
of $x$ being selected in the $i^{th}$ round conditionally depends on $x$ not being selected in previous rounds\footnote{Y follows a hypergeometric 
distribution with $2k$ population size, $k$ draws and 1 observed and possible success state. Thus, the following could also be derived 
straightforwardly using the probability mass function for the distribution.}. Then, 

\begin{center}
 $Pr[Y = Y_i] = Pr[Y_i | Y \neq Y_{i-1},Y_{i-2},\ldots,Y_1]$
\end{center}

Note that for round $i$, an item is selected from the combined set with probability $1/(2k - i)$ since previously selected items are removed without replacement.

Now, $Pr[Y = Y_1] = 1/2k$.

\begin{center}
$Pr[Y = Y_2] = Pr[Y_2 | Y \neq Y_1] = 1/(2k-1) \times Pr[Y \neq Y_1] = 1/2k$.
\end{center}

It can be similarly shown that $Pr[Y = Y_i] = 1/2k$ $\forall i \in k$. 
Thus, 

\begin{center}
$Pr[Y] = \sum\limits_{i = 1}^{i=k}Pr[Y = Y_i] = 1/2$.
\end{center}

Finally, $Pr[X^{'}] = Pr[Y] \times Pr[E_1 U E_2] = 1/2 \times 2k/N = k/N$, thus proving $Pr[X^{'}] = Pr[X]$.
\end{proof}


\smallskip
\noindent
{\bf Stash.~}
If none of the $k$ rounds yields a free block, 
the data is written to an in-memory stash. Since, initially there are equal number 
of items from the {\em free set} and the {\em random set} in the 
combined set and uniformly random items are selected each round, $k$ rounds of the protocol yields an expected $k/2$ number of 
items from the {\em free set}.
If $k = 3$, the \sysoram~ stash can be modeled as a {\em D/M/1} queue 
similar to \cite{Blass:2014:TRH:2660267.2660313}, and bound 
to a constant size with negligible failure probability in the security 
parameter. We refer to \cite{Blass:2014:TRH:2660267.2660313} for details.

\smallskip
\noindent
{\bf Device utilization.~}
A final detail is to ensure that the expected 
number of free blocks obtained from the selection protocol does not exceed the 
expected number of free blocks obtained from an equal-sized sample selected randomly out of 
all the blocks in the ORAM. 
Otherwise, the protocol will always yield an expected 
$k/2$ free blocks unlike a randomly selected sample where the expected number 
of free blocks will be a function of the actual distribution of data. Note that 
lemma~\ref{lemma_rand_ind} shows indistinguishability between a randomly selected sample and 
the sample obtained due to the modification in Step 2 of the protocol for {\em each} write. However, 
a significant difference in the expected number of free blocks obtained can leak subtle correlations 
over {\em multiple} rounds. Fortunately, a straightforward solution for this is to 
ensure that half of the ORAM blocks are always free. Note that a similar assumption is also 
made by HIVE~\cite{Blass:2014:TRH:2660267.2660313} but for a different purpose, namely to bound the 
in-memory stash for the write protocol.

\begin{lemma}
If half of the ORAM blocks are always ensured to be free, the expected number of free blocks provided 
 by the selection protocol per write is equivalent to the expected number of free blocks in a  
 sample of $k$ blocks randomly selected out of all $N$ blocks in the ORAM.
\end{lemma}

\begin{proof}
 Note that a free block can be obtained by the protocol only if an item is selected from the {\em free set} in Step 1, for 
 at least 1 out of the $k$ rounds.  Since, the rounds are independent, it can be straightforwardly 
 observed that the expected number of free blocks yielded 
 by the selection protocol is $k/2$. This is equal to the expected number of free blocks 
 in a randomly selected sample of $k$ blocks, if half of the blocks in the ORAM are always free.
\end{proof}

\smallskip
\noindent
{\bf Access Complexity.~}
Both $read\_oram()$ and $write\_oram()$ access the ORAM map to locate the target
block. The complexity of accessing an entry in the B+ tree is $\order{log_{\beta}N}$.
Further, $write\_oram()$ needs to find and write to $\order{log_{\beta}N}$ 
free blocks for writing data and updating 
the map blocks. Finding a free block requires $\order{1}$ accesses as discussed above,
and thus the write complexity of \sysoram~ is $\order{log_{\beta}N}$. Consequently, the overall access complexity of \sysoram~  is 
$\order{B\times log_{\beta}(N)}$.

%
%


\subsection{Security Over Multiple Rounds}
\label{oram_pets_fix}
\begin{figure}
\centering
 \includegraphics[scale=0.15]{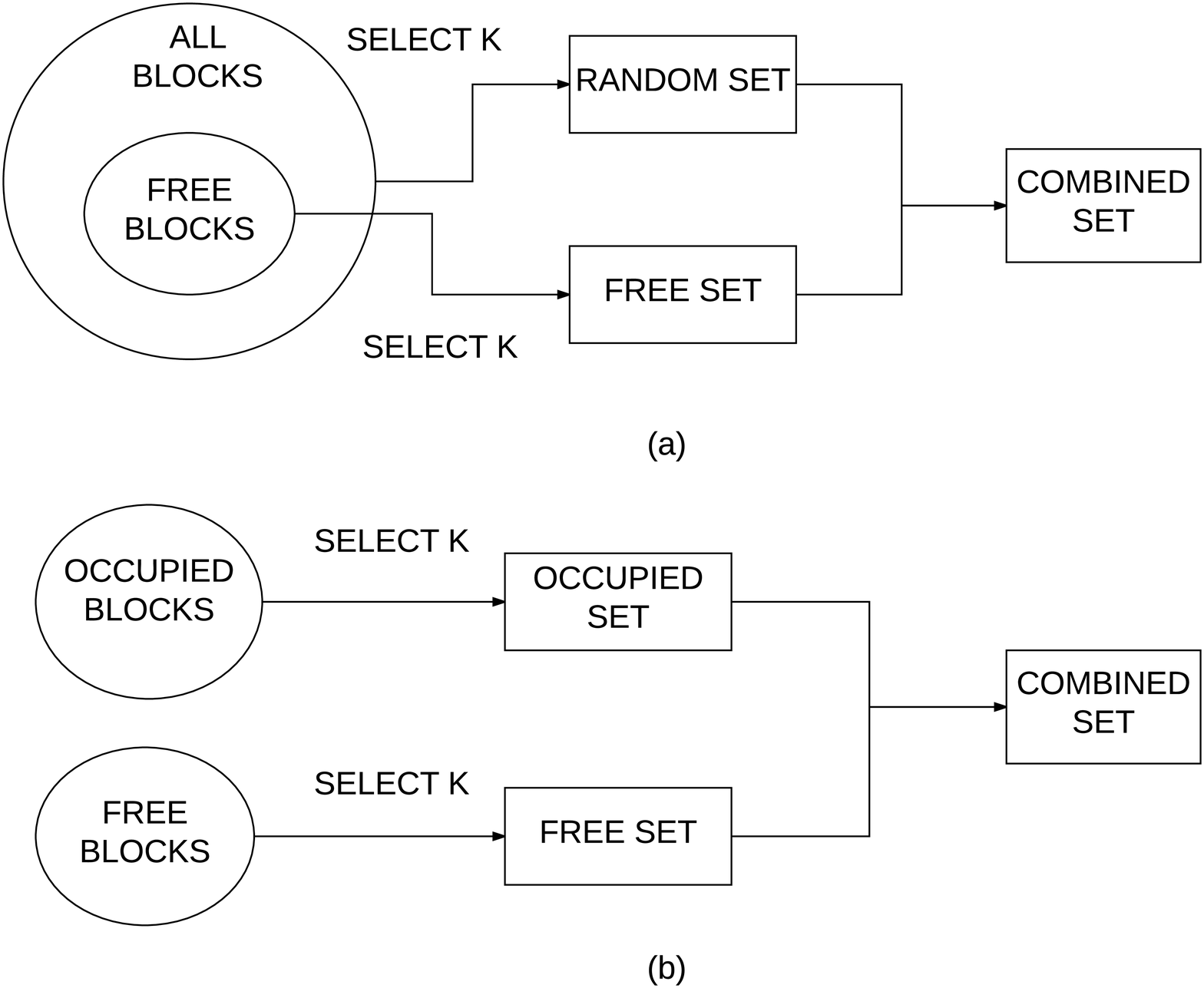}
\caption{Free block selection. The free block selection described in (a) is biased towards selecting 
more free blocks into the combined set since the set of free blocks is a subset of all blocks on the disk, 
which in turn is used for selection into the random set. A non-biased selection mechanism (b) 
selects equal number of block out of the occupied blocks and the free blocks into the 
combined set. The set of occupied blocks and free blocks are disjoint by construction.\label{selection_fig} 
}
\end{figure}

As per a recent email conversation with Roche {\em et al.} \cite{detworam}, an adversary 
can mount an attack on \sysoram~ over 
multiple rounds of execution based on the observation that the
free block selection mechanism described in Section \ref{ORAM:findingfree}
is slightly biased towards selection of free blocks (Figure~\ref{selection_fig}(a)). In particular, 
the set of all free blocks is a subset of all blocks on the disk. Since, an 
equal number of elements is selected into the combined set from the
free set (sampled uniformly out of the set of free blocks) 
and the random set (sampled uniformly out of all blocks on the disk), a free block is more likely to be selected into 
the combined set than an occupied block. In fact, the adversary's advantage in guessing whether a block is free or occupied, 
over multiple rounds of observation is $\order{1/N}$, which is non-negligible in the security parameter. However, note 
that for practical purposes, $1/N = \order{2^{-30}}$ is still significantly small for large storage devices. 

In order to mitigate this security leak, we propose a new selection protocol (Figure~\ref{selection_fig}(b)),
ensuring that all blocks are equally likely to be in the combined set, irrespective of their 
state of occupancy. More specifically, the combined set is now formed with equal number of elements 
uniformly sampled from the set of all free blocks and the set of all occupied blocks. Note that 
these sets are disjoint and are of equal size when the disk is full -- since \sysoram~ allows only 
half of the block to be occupied at any given time. Because the sets are disjoint, all blocks 
are equally likely to be selected to the combined set from the set of free blocks {\em or} 
the set of occupied blocks. 

\smallskip
\noindent
{\bf Initialization.~}
For this mechanism to work, the set of free blocks and the set of occupied blocks should always be of equal size. This 
can be easily achieved by running an {\em initialization} phase where half of the disk is 
written with random data. In this case, all future accesses to the disk will be substitutions 
of the {\em logical} blocks written during the initialization since \sysoram~ allows only 
half of the disk to be occupied. Note that an initialization step is a standard assumption for ORAM protocols
which assume that logical data (possibly random) is loaded into the ORAM to setup an initial 
configuration of associated data structures (such as the position map etc. \cite{pathoram}).

Deploying the new selection mechanism requires several changes to the scheme as we need to 
maintain the block addresses of all occupied blocks efficiently. This is achieved 
using a {\em non-free block matrix} (N-FBM), a structure similar to the FBM.

\smallskip
\noindent
{\bf Non- Free Block Matrix (N-FBM).~}
As discussed, the N-FBM is used to store  
block addresses of all currently occupied blocks 
at random locations. The N-FBM is an encrypted $\beta \times N/\beta$  matrix (where $N$ is the total number of
physical blocks in the ORAM and $\beta$ is number of physical block addresses that can be 
written to one disk block). 
The columns of the matrix are stored in {\em consecutive}
disk blocks outside the ORAM, each containing $\beta$ free block addresses (Figure
\ref{fbm_design}). The coordinates in the matrix where each block address is 
stored is randomized independent of the address. 
Contrary to the FBM, the N-FBM is empty at initialization. 
During the initialization step described above, the block 
addresses corresponding to the physical blocks being written 
are added to random locations in the N-FBM. Thus, 
after the initialization step, the N-FBM contains 
block addresses of all currently occupied blocks 
at $N/2$ random locations.

\begin{figure}
\centering
 \includegraphics[scale=0.15]{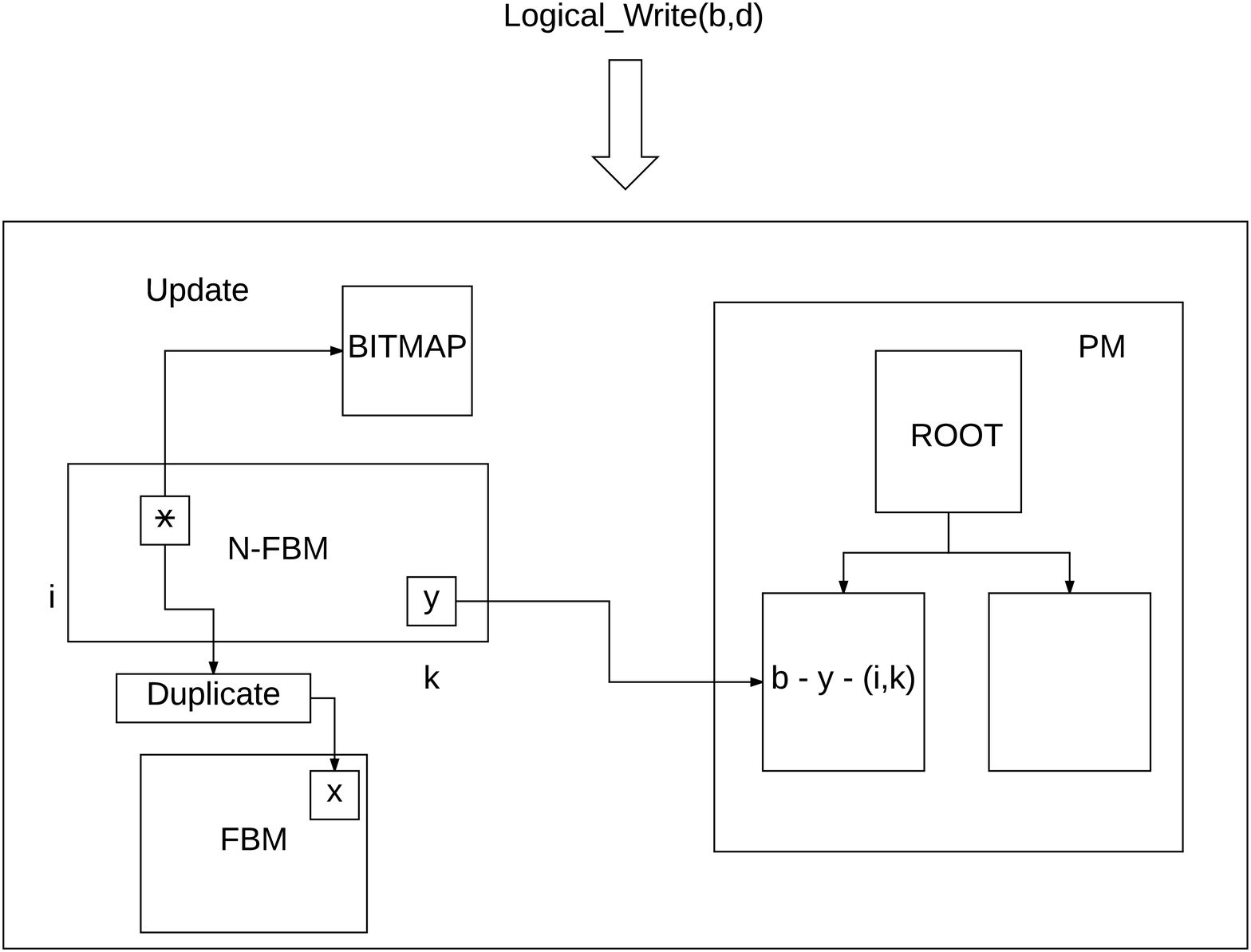}
\caption{The FBM and the N-FBM synchronization. On a logical write to block $b$ originally located at 
physical block $x$ on disk, a new free block $y$ is first selected from the FBM. $y$ is replaced with 
$x$ in the FBM followed by update to the position map (PM). The PM entry for $b$ is updated with location 
where $y$ is added to the N-FBM. Additionally, the PM entry for $b$ before this access contained 
the location where $x$ was added in the N-FBM previously. Once $x$ is added back to the 
FBM, its previous location in the N-FBM is marked as free in the bitmap to indicate a duplicate entry.\label{sel_desc}}
\end{figure}

\small
\noindent
{\bf Adding addresses to the N-FBM .~}
The challenge is to keep oblivious update the N-FBM, i.e. 
addresses corresponding to physical blocks used for writes after the initialization need to be added to the N-FBM securely. 
A naive solution of adding block addresses to deterministic locations in the N-FBM 
will allow the adversary to link these modifications and thus violate security. 
Consequently, block addresses have to be added to random locations thus generating random write traces to the N-FBM.

Fortunately, only $N/2$ locations in the N-FBM can contain valid physical block addresses 
of currently occupied blocks since only $N/2$ physical 
blocks out of all $N$ blocks on the disk can be occupied at any given time. 
The rest of the $N/2$ locations in the N-FBM 
are ``free'' and new block addresses can be added there. Consequently, a randomly selected location from the N-FBM 
will be free with probability $1/2$.  \sysoram~ tracks the locations in the N-FBM that do not contain
valid block addresses of currently occupied blocks using a bitmap (assumed to be in memory, 
later we show how to store on disk). More specifically, the bitmap tracks the locations 
in the N-FBM that are currently free.

Then, to find a free location in the N-FBM, a location is randomly selected and its corresponding bitmap entry is 
checked to determine its state of occupancy. In case, the location selected is not free, a new random location is selected
and the process is repeated until a free location is obtained. With security parameter, $\lambda$, 
the probability that a free location is {\em not} obtained even after repeating the process $\lambda$ times 
is $\order{2^{-\lambda}}$, which is negligible in $\lambda$. We however note 
that this  process is expensive and show how 
to perform this efficiently after describing our free block selection protocol.

\small
\noindent
{\bf Handling duplicates.~}
An obvious problem to adding block addresses to random free locations is that 
the N-FBM soon fills up with duplicate entries. Consider the following scenario 
where logical block $b$ is stored at physical block $y$. When $b$ is updated by a future access, 
it is moved to a new physical block $x$ (selected from the FBM) while $y$ is now free and 
can thus be added back to the FBM (Section~\ref{ORAM:findingfree}). Additionally, $x$ is added to a random location in 
the N-FBM. Note that although $y$ is free, an entry for $y$ still exists in the N-FBM as 
it must have been added on the previous access to $b$! This results in duplicate entries 
in the FBM and the N-FBM which breaks the assumption that the two sets are disjoint. 

In order to rectify this, \sysoram~ updates the N-FBM bitmap accordingly to indicate that the location 
previously occupied by $y$ is now free. \sysoram~ keeps track of this location using the 
position map. More specifically, when $y$ is added to the N-FBM, as part of a write to $b$, 
\sysoram~ additionally stores the location where $y$ is added to the N-FBM in the position 
map entry for $b$. Subsequently, when $y$ is freed (next access to $b$), \sysoram~ updates the bitmap to indicate 
that the corresponding location in the N-FBM is free. 

Figure \ref{sel_desc} describes the mechanism of adding blocks to the FBM and the N-FBM during free 
block selection.

\smallskip
\noindent
{\bf Free Block Selection.~}
The new free block selection protocol (Figure \ref{selection_fig} (b)) can be implemented straightforwardly using the 
FBM as the set of free block addresses and the N-FBM as the set of 
occupied block addresses as follows --

Creates two sets, each with $k$ items as follows -- 

\begin{enumerate}
\item A {\em free set} that contains $k$ randomly chosen free block addresses from the FBM.
\item A {\em occupied set} that contains $k$ randomly 
chosen block address from the N-FBM .  
\end{enumerate}

Combine the two sets into a {\em combined} set. Note that since the two sets are disjoint, 
there are no duplicates in the combined set.

\begin{enumerate}
\item Select an item randomly from the combined set.
\item For the item, $i$ selected in step 1 --
\begin{enumerate}
 \item If $i$ originally belonged to the {\em free set}, 
use the corresponding block for the
 ORAM write and update the FBM and N-FBM.
 \item If $i$ is from the occupied set, reencrypt the block corresponding to the 
 address and reencrypt an additional 
random location from the N-FBM. 
\end{enumerate}
 Remove $i$ from the combined set.
\item Repeat steps 1-3, $k$ times.
\item For all the remaining items, $i$, that are also in the {\em free set}, 
replace the addresses back in the FBM from where they were selected. 
\end{enumerate}

\begin{lemma}
\label{lemma_rand_ind_fix}
The sample of $k$ physical blocks modified for {\em every} ORAM write
(due to the selection protocol) is indistinguishable from a sample of $k$ blocks 
selected uniformly at random out of the $N$ physical blocks in the ORAM.
\end{lemma}

\begin{proof}
 First, note that for a block $x$ that is currently, the probability of being selected to 
 the combined set is $Pr[x \in combined | x \in free ] =  1/2 \times k/(N/2) = k/N$, i.e.  
 probability that $x$ is selected into the combined set depends on the probability 
 that $x$ is free (=1/2 since at any given time $N/2$ random blocks are free out of all $N$ blocks on the disk) 
 and that $x$ is one of $k$ selected block addresses from the FBM.  
 Similarly, the probability of $x$
 being selected to the combined set if $x$ is occupied is $Pr[x \in combined | x \in occupied ] = 1/2 \times k/(N/2) = k/N$. 
 Thus, $Pr[x \in combined] = 2k/N$. Now, if $x$ is in the combined set, the probability that $x$ will 
 be one of the $k$ items selected randomly out of the $2k$ items in the 
 combined set is $Pr[x \in k | x \in combined] = 1/2 \times 2k/N = k/N$. 
 This is equal to the probability, $P[x \in random] = k/N$ of $x$ being a part of any random sample 
 of $k$ blocks selected uniformly out of all $N$ disk blocks.
 
\end{proof}

\smallskip
\noindent
{\bf Selecting free locations in the N-FBM. Bounding the stash size.}
As discussed before after a logical block has been written to a physical block, 
the physical block address needs to be added to a random location in the N-FBM 
that does not contain a valid block address of a currently occupied block -- that is a 
free location. To achieve this, random locations from the N-FBM are selected until 
a free location is obtained to write the block address. Unfortunately, 
this is an expensive process

An alternative solution is to write the logical block to the in-memory stash when 
randomly selecting a location from the N-FBM does not yield a free location to 
write the corresponding physical block address. For example, 
when a logical block $b$ is to be updated, and the free block selection protocol 
provides physical block $y$ to perform the write, \sysoram~ randomly selects a 
location from the N-FBM. If the location is free (as determined by the bitmap), 
$b$ is written to $y$, and $y$ is written to the selected location in the N-FBM.
On the other had, if the location selected from the N-FBM is not free, $b$ is 
written to the stash and $y$ is added back to its original location in the FBM.
Further, the block address already stored at the N-FBM location selected is reencrypted.
In effect, there is no change to the state of the FBM or the N-FBM.




 
However, in this case, the logical block stash size 
needs to be bounded under the new constraints. Note that now the probability that each block write succeeds 
depends on two conditions -- the probability that the random free block selection protocol yields a free block in round $i < k$ of the protocol (= 1/2) {\em and} 
the probability that a randomly selected location from the N-FBM is free (=1/2). Thus, the probabilty with which a block is successfully 
written to the disk in one round of the selection protocol is 1/4. Then, to bound the stash, the value of $k$ needs to be modified such that 
the rate parameter of the D/M/1 queue (on which the stash is modelled), $\mu = k/4 > 1$. This
ensures that the stash size is bounded to a constant size with failure probability negligible in the security 
parameter \cite{Blass:2014:TRH:2660267.2660313}. \sysoram~ sets $k = 5$.


\smallskip
\noindent
{\bf Storing the bitmap.~}
A final detail is to securely store the bitmap. Since, large DRAMs are ubiquitous in modern systems, 
the bitmap could be stored entirely in memory and thus does not need to be subject to the ORAM 
protocols. For example, for a 1TB storage device, 32MB of memory is enough to hold the bitmap. 
Nonetheless, for large storage devices it may be
desirable to store the bitmap on the disk to ensure consistency accross failures.
Note that if the bitmap is stored on disk, it also needs to be subject to the ORAM protocols 
to ensure that the adversary cannot link accesses based on the bitmap updates. More specifically, 
when a part of the bitmap is updated, it needs to be relocated to a new random location.

The bitmap is stored in multiple random disk blocks while an array (``bitmap position array'') 
in memory tracks the location of the disk blocks storing a part of the bitmap. The ``bitmap 
position array'' is significantly smaller in size than the \sysoram~ stash. For example with 4KB physical block size, the 
number of disk blocks required to store the bitmap for a 1TB device is $2^{13}$. With the  bitmap position array storing 
a 8 byte physical block address for each such disk block, the total size of the array is 64KB. Storing the bitmap thusly 
ensures that accessing the bitmap requires only 1 disk access. 

\smallskip
\noindent
{\bf Asymptotic Complexity.~}
The new free block selection protocol and the addition of the N-FBM and the bitmap does not 
change the asymptotic complexity from Section \ref{ORAM:findingfree}. First, selecting a {\em non-duplicate} 
block address from the N-FBM to the occupied set requires only a constant
number of accesses within which at least one expected 
{\em non-duplicate} block address in the N-FBM will be found. If the bitmap is stored 
in-memory, the accesses are free. Even when the bitmap is stored on disk, 
accesses to a part of the bitmap can be performed with 1 disk access as described before. 
The FBM requires constant number of accesses for adding/selecting items. Thus, the 
free block selection presented here can be performed with $\order{1}$ access complexity and the ORAM write complexity 
is dominated by the update to the position map -- \sysoram~ write complexity is still $\order{log_{\beta}N}$

\begin{thm}
\label{thm_security_oram}
 \sysoram~  provides write access pattern indistinguishability.
\end{thm}

\begin{proof}
To show write access pattern indistinguishability, we show the 
existance of a simulator that can generate a write trace indistinguishable from the write 
trace generated by {\em any} \sysoram~ write, without any knowledge
of the block being written {\em or} of the data stored in the ORAM, 
based only on public information. 

First, note that each ORAM write (including updates to the position map) 
invariably generates the following write trace -- 

\begin{enumerate}
 \item Modify $k$ {\em unique} disk blocks. 
 \item Modify $k$ locations in the FBM.
 \item Modify $k$ locations in the N-FBM.
 \item Modify $2 \times k$ bitmap locations (if stored on disk).
\end{enumerate}

Now a simulator, $S$ can generate a write trace indistinguishable from the 
write trace described above as follow --

\begin{enumerate}
 \item Reencrypt $k$ {\em random} disk blocks. As shown by Lemma~\ref{lemma_rand_ind_fix}, the $k$ blocks 
 modified by the write are indistinguishable from $k$ blocks selected randomly out 
 of all $N$ blocks.
  \item Reencrypt $k$ {\em random} locations from the FBM. Due to uniformly random selection of entries from 
  the FBM, this is indistinguishable from a real FBM access.
 \item Reencrypt $k$ random locations in the N-FBM. Due to all entries in the N-FBM being at 
 uniformly random locations, this is indistinguishable from a real N-FBM access. 
  \item Reencrypt $2 \times k$ {\em random} bitmap locations. As random locations in the N-FBM are modified (due 
  to Step 3), this is indistinguishable from a real modification to the bitmap.
\end{enumerate}

Semantic security of the encryption scheme ensures that reencrytion is indistinguishable from a real 
modification. Finally, each access results in writes to
blocks of the position map along with the block that 
the user wants to access. Position map and data blocks are re-encrypted with semantic security and thus 
indistinguishable. 

%
\end{proof}

\smallskip
\noindent
{\bf Simulating a write.~} 
Theorem \ref{thm_security_oram} shows that a $write\_oram$ write trace can be generated without 
any knowledge of the write being performed or the data stored in the ORAM. Consequently, a 
$write\_oram$ access can be simulated even with random data stored in the ORAM and the auxiliary 
data structures (FBM, N-FBM, bitmap etc.).  We use this property of \sysoram~ as a basis for 
the solution described in Section \ref{cons}. 
%
Due to semantic security, re-randomizing blocks with random data is indistinguishable from reencrypting a block with valid data.

\section{Access Type Indistinguishability}
\label{cons}
\begin{figure}
\centering
 \includegraphics[scale=0.20]{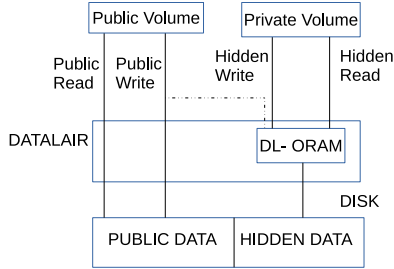}
\vspace{-0.5cm}
 \caption{Design of DataLair {\em Lite}. Writes to the hidden partition
  happen through DL-ORAM.\label{datalair_basic}}
\end{figure}

We now detail the design of \sysname~, which ensures HWA by
deploying \sysoram~ (Section~\ref{ORAM}) and PAT through its access
protocols described below.  For PAT, \sysname~ ensures that a device
containing both public and hidden data is indistinguishable from a device
containing only public data.

First, we describe a simple secure design, \sysname~ {\em Lite}, which
sacrifices storage space for reduced design complexity.  Section
\ref{cons_final}, introduces a more complex but space-efficient design
dubbed \sysname~ {\em Full}.

\subsection{\sysname~ Lite: Isolating Volumes}
\label{cons_initial}
\noindent
{\bf Setup.~} 
\sysname~ Lite maps blocks from {\em two} logical volumes to an underlying
block device and can be set up in two modes of operation -- $ONLY\_PUB$
(only public) and $PUB\_HID$ (public and hidden).  In the $ONLY\_PUB$ mode,
the device only stores data from a ``public'' logical volume that is
disclosed to an adversary.  For $PUB\_HID$ mode, \sysname~ Lite also stores
data from a hidden volume.  \sysname~ Lite fixes the size of each volume to
50\% of the underlying physical device size.  Each logical volume may
support a filesystem.  \sysname~ Lite creates two physical partitions on the
device -- a public partition and a hidden partition, according to the
corresponding logical volume sizes (Figure~\ref{datalair_basic}).  When,
\sysname~ Lite is initialized in $ONLY\_PUB$ mode, the hidden partition is
filled with random data.

The underlying storage is a block device with $N$ blocks of size $B$ each
and physical block address $P_{id} \in [1,N]$.  The format of the address
would vary across different types of block devices, e.g., in case of a
hard disk, the physical block address would be the sector numbers that constitute
one block on the physical device.

Logical block addresses $V_{id} \in [1,|V_{id}|]$ are used to reference blocks in
the logical volumes ($|V_{id}|$ is the size of the volume).  The logical
blocks are mapped to the physical blocks using a device mapper.  Data from
the public volume is mapped to the public partition directly (as 
indicated by the overlying FS) while data from
the hidden volume is mapped to the hidden partition using \sysoram~.  In
addition, semantically secure encryption is used to encrypt all data and
metadata before writing.  Public data is encrypted with a key $K_{pub}$
available to the adversary.  Hidden data is encrypted with secret key
$K_{hid}$. In practice, keys may be derived from user passwords.

\smallskip
\noindent
{\bf I/O.~}
\sysname~ Lite maps logical volume I/O into either public or hidden volume operations.   

Public reads and writes are straightforward since \sysname~ Lite can
linearly translate the logical block address and perform a read/write to the
corresponding physical block in the public partition.  A read to the hidden
volume calls $read\_oram$ on the hidden partition.  In addition, to ensure
PD, with {\em every} public write, \sysname~ Lite performs a hidden
operation as follows:

\begin{itemize}

\item If there is a write queued up for the hidden logical volume, it is
performed by calling $write\_oram$.

\item In the absence of a write for the hidden volume, or if \sysname~ Lite
is used in $ONLY\_PUB$ mode, a write is {\em simulated} for \sysoram~ (as
described in Section~\ref{ORAM}).

\end{itemize}
 
Effectively, this ensures that every write to the hidden volume is preceded
by a write to the public volume.  If there is no public write when a hidden
write request arrives at \sysname~ Lite, the hidden data is queued in the
\sysoram~ stash.  With every write to the public volume, either an
outstanding write (or data from the stash) is written to the hidden volume
(using \sysoram~) or a write is {\em simulated}.  As shown before, these two
cases are indistinguishable to an adversary without the key for \sysoram~.

\smallskip
\noindent
{\bf Security.~} 
The construction described above provides PD for the hidden volume.  First,
both the $ONLY\_PUB$ and $PUB\_HID$ modes of operations create public and
hidden partitions of equal size.  A write to the public volume in both the
cases results in indistinguishable modifications to the hidden partition --
either due to an actual write or a simulation.  This guarantees PAT
indistinguishability.  Further, deploying \sysoram~ on the hidden partition
ensures HWA indistinguishability.  Recall that these are the necessary and sufficient
conditions to ensure PD-CPA security.

\subsection{\sysname~ Full: Merging Volumes}
\label{cons_final}
Although, \sysname~ Lite (Section~\ref{cons_initial}) achieves PD, it makes
sub-optimal use of storage space -- e.g., in $ONLY\_PUB$ mode, a
hidden partition uses up 50\% of the space allocated to it, notwithstanding
of actual use.  A more reasonable solution would allow physical volume
storage space to correspond to logical use requirements.  Further, space not
used for hidden data should be available for public data and vice-versa.

To this end, \sysname~ Full allows the user to create two (or more) volumes
of variable sizes and stores them on the {\em same} physical partition.  In
this case, both the public and the hidden volume can be of the same logical
size as the underlying partition and use all the available space (in this
case up to 50\% of total device size) for either hidden or public data. 
We provide the intuition for the restricted device usage further below.

Unfortunately, achieving this is significantly more challenging than the
{\em Lite} construction -- the main problem being mapping public data
to independent locations in the presence of hidden data.  We detail below.

\smallskip
\noindent
{\bf Mapping Public Data.~}
\label{final:mapping_public}
First, with both public and hidden data being stored within the same
physical address space, writing public data straightforwardly to physical blocks 
indicated by an overlying filesystem is not possible. Since \sysname~ does
not restrict the choice of filesystem, 
the distribution of
public data will also determine the distribution of hidden
data.  For example, with a log structured filesystem on the public volume,
all hidden data will end up being ``pushed" towards the end of the disk.  This
breaks the security of the random free block selection mechanism in
\sysoram~.  Instead, public data will need to be mapped {\em
randomly} without compromising
overall PD. This requires storing a corresponding mapping table for the public volume .

\smallskip
\noindent
{\bf Public Position Map (PPM).~}
\sysname~ Full stores the logical to physical block address mappings for public
data in an array called the {\em Public Position Map} (PPM), stored at a fixed
device location.  The PPM is similar to the mapping table used by most
device mappers.  Importantly, the PPM is considered to be public data and
thus not subject to PD.

To proceed, it is necessary to define two important terms here to 
categorize physical blocks on the basis of their state of occupancy: 
{\em truly free} and {\em apparently free}.
{\bf Truly free block:~} a block that does not contain any (public or hidden) data.
{\bf Apparently free block:~} a block that contains hidden data and the use
of which needs to be hidden from an adversary, i.e., the block needs to
``appear'' free to an adversary.

To maintain PD, public data writes should not avoid {\em apparently free}
blocks by writing {\em around} hidden data.
Thus, while writing public data to randomly selected blocks, \sysname~ must 
ensure that {\em all} free blocks (including {\em truly free} and {\em apparently free}) are equally 
likely to get selected to complete the write.
An obvious solution then is to choose a
random block and write the public data there if it is unoccupied. 
If the block is {\em apparently free}, the hidden data there can be relocated
to a new random location subsequently.

This approach however creates a significant problem.  Recall that \sysoram~
writes hidden data by deploying the uniform free block selection protocol (detailed
in Section~\ref{oram:free_block}), using the FBM for selecting ``free'' block
addresses.  In the current context, to ensure correctness, 
the FBM should contain addresses of only {\em truly free} blocks, i.e., blocks
that do not contain either public or hidden data.  Thus, randomly choosing a block for 
writing public data will also require the
corresponding block address to be invalidated in the FBM.  Otherwise, the FBM will contain
block addresses that are already occupied by public data. If such a block address 
is selected for a subsequent hidden write, the 
public data in the block will have to be moved elsewhere (leaking the 
presence of hidden data) to complete the write. 
Unfortunately, invalidating a particular randomly chosen block address in the FBM is not
straightforward -- by construction the location of a block address in the FBM is
randomized.

The solution then is to select block addresses for public writes also 
from the FBM while updating the FBM in the process. The problem however with 
naively implementing this is that the FBM is hidden data i.e, the FBM is encrypted 
with the \sysoram~ secret key. Thus, using the FBM for public writes 
would entail the user to provide the FBM key to the adversary since 
all public operations and data structures in \sysname~ needs to be transparent 
to the adversary for PD. Providing the FBM key to the adversary 
breaks the security of \sysoram~ if hidden data is also stored.
\begin{figure}
\centering
\includegraphics[scale=0.15]{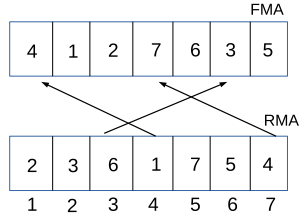}
\vspace{-0.2cm}
\caption{Sample PFL with 7 entries. Each entry in the FMA is a physical block address. The reverse mapping array (RMA) 
stores an index pointer to 
the forward mapping array (FMA) for a particular entry.\label{pub_free_list}}
\end{figure} 

\smallskip
\noindent
{\bf Public Free List (PFL).~}
To solve this, \sysname~ selects free block addresses from the FBM but also provides a 
way to {\em plausibly deny} this. The idea is to store a public (encrypted with the public key) 
list of block addresses 
corresponding to the blocks that do not contain public data. The list further supports 
the following two properties: i) it allows efficient retrieval of uniformly random entries, and 
ii) the location of any given {\em entry} in the list can be determined efficiently. We detail below.

The PFL (Figure~\ref{pub_free_list} ) is a data structure keeping track of block addresses of blocks
that {\em do not} contain public data.  The PFL is public (not subject to deniability) 
and is composed of
two arrays:

\begin{enumerate}
\item The {\em forward mapping array} (FMA) (size $N$) of block addresses that
currently do not contain public data  ({\em apparently free} + {\em truly free}). 
Uniformly random block addresses can be selected by picking up the entry at a
randomly selected index in the array.  If the selected entry is to be
removed from the PFL (use case described later), the array is subsequently
{\em compacted}, by moving the entry at the end of the array to the index
corresponding to the removed item.  The compacting ensures that the {\em
real} size of the array is always known and entries can be picked uniformly.

\item The {\em reverse mapping array} (RMA) (size $N$) tracks the index
in the forward mapping array corresponding to each physical block address. 
Whenever an entry is removed/added or replaced in the forward mapping array,
the reverse mapping array is updated as well.
\end{enumerate}

The intuition behind the PFL is that the FMA allows \sysname~ 
to efficiently select and retrieve a free block address that does not contain 
public data, while the RMA allows efficiently locating 
the index of a particular free block address in the forward mapping array. 

\smallskip
\noindent
{\bf Public write free block selection.~}
\sysname~ Full performs the public write free block selection as follows. First, 
the uniform free block selection protocol using the FBM is deployed as described in Section~\ref{oram:free_block}.
More specifically, for a public write, \sysname~ Full creates a {\em free set} and the {\em random set}. 
In addition, now the {\em random set} can be built by using random block addresses from the PFL because 
the blocks that already contain public data (and thus not part of the PFL) 
can be trivially excluded as being occupied.
This is followed by executing the $k$ rounds of the protocol using the combined set. 
If the 
protocol yields a free block address then the data is written to that block. 
Then, using the RMA, \sysname~ Full determines the location of the block address 
in the FMA, which is subsequently removed. 
This allows \sysname~ Full to claim that the 
block address was actually selected from the PFL. 

If the protocol yields no free blocks, the data is
still written to the disk (instead of being added to the in-memory stash). The idea here is to ensure that 
a public write always translates to a write to the disk. Since the modifications 
to blocks not containing public data (due to a hidden write or \sysoram~ simulation) 
are attributed to public data writes, writing 
public data to the stash can result in inconsistent modifications on the disk and 
violate PD.

To implement this, when the selection protocol does not yield 
a {\em truly free} block for a public write, 
the public data is instead written to an {\em apparently free} block. 
The hidden data there can then be moved to the stash and written back in a subsequent
hidden write.  An {\em apparently free} block address can be straightforwardly
selected from the \sysoram~ position map, which as described before 
stores the logical to physical address mappings for data 
in the ORAM (hidden data in this case). Subsequently, this 
block address is removed from the PFL. Note that the PFL necessarily contains 
this address, since it contains entries of {\em all} blocks not 
containing public data -- which also includes blocks that already contain 
hidden data ({\em apparently free} blocks). Thus, a free 
block address selected using this procedure can also be attributed to being 
selected from the PFL


\smallskip
\noindent
{\bf Hidden write free block selection.~}
Hidden writes follow the same procedure as \sysname~ Lite by invoking 
the \sysoram~ write protocol. However unlike public writes, 
hidden data is still written 
to the stash if 
the selection protocol yields no free blocks. Recall 
that one of the requirements for a bounded \sysoram~ stash is to
ensure that half of the ORAM is free.  In this case, since \sysoram~ will
write to the blocks which are shared with public data, it necessary to
ensure that the combined size of public and hidden data is only half the
size of the device.  This is achieved by the \sysname~ device mapper only
allowing the user to create logical volumes with size equal to or less than
half of the device capacity.  Further, on reaching 50\% utilization, the
device mapper informs the user that the disk is full. 

\begin{figure}
\centering
 \includegraphics[scale=0.18]{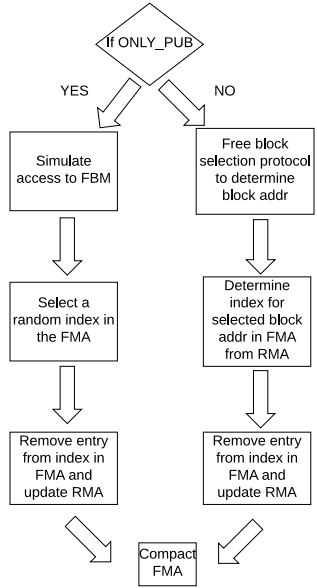}
 \vspace{-0.3cm}
 \caption{Free block selection protocol for public writes in \sysname~ Full. Step 1 is indistinguishable 
 for the two cases since simulation to the FBM is indistinguishable from a real access. 
 Using either the PFL or the free block selection protocol in step 2 provides uniformly random block addresses 
 of blocks that do not contain public data. The entry selected by the selection protocol in step 2 
 will be at a random index in the FMA as all entries in the FBM are randomized independent of their 
 locations in the FMA.\label{write_summary}}
\end{figure}

\smallskip
\noindent
{\bf Indistinguishability between two modes of operation.~}
Recall that in the $ONLY\_PUB$ mode, \sysoram~ and
the FBM is initialized with random data. Thus, the free block selection mechanism 
for public data writes described above (using the FBM) cannot be deployed in this case.
Fortunately, to overcome this the PFL can be used to efficiently select
uniformly random free block addresses. Once the address has been selected and removed 
from the FMA, the array is compacted.

Further, to ensure indistinguishability, the accesses due to the uniform free 
block selection using the FBM is simulated by reencrypting the required number 
of FBM blocks and required number of random block addresses selected from the PFL.

In summary (Figure~\ref{write_summary}), 
when writing public data in $ONLY\_PUB$ mode, a randomly chosen
block address is removed from the PFL while simulating the uniform free block selection protocol.
In $PUB\_HID$ mode, the randomly chosen block address for writing
public data is determined using the free block selection protocol while removing its corresponding address
from the PFL.  The two cases are indistinguishable since an actual access to
\sysoram~ and the FBM is indistinguishable from a simulation.

\smallskip
\noindent
{\bf Storing metadata for encryption.~}
Since encryption is performed at the block level and the reverse mapping array contain an entry 
for each disk block, the IVs/counters for the randomized semantically secure cipher used to
encrypt the physical blocks are stored in the reverse mapping array. 


%
%
%
%

\begin{figure}
\center
 \includegraphics[scale=0.20]{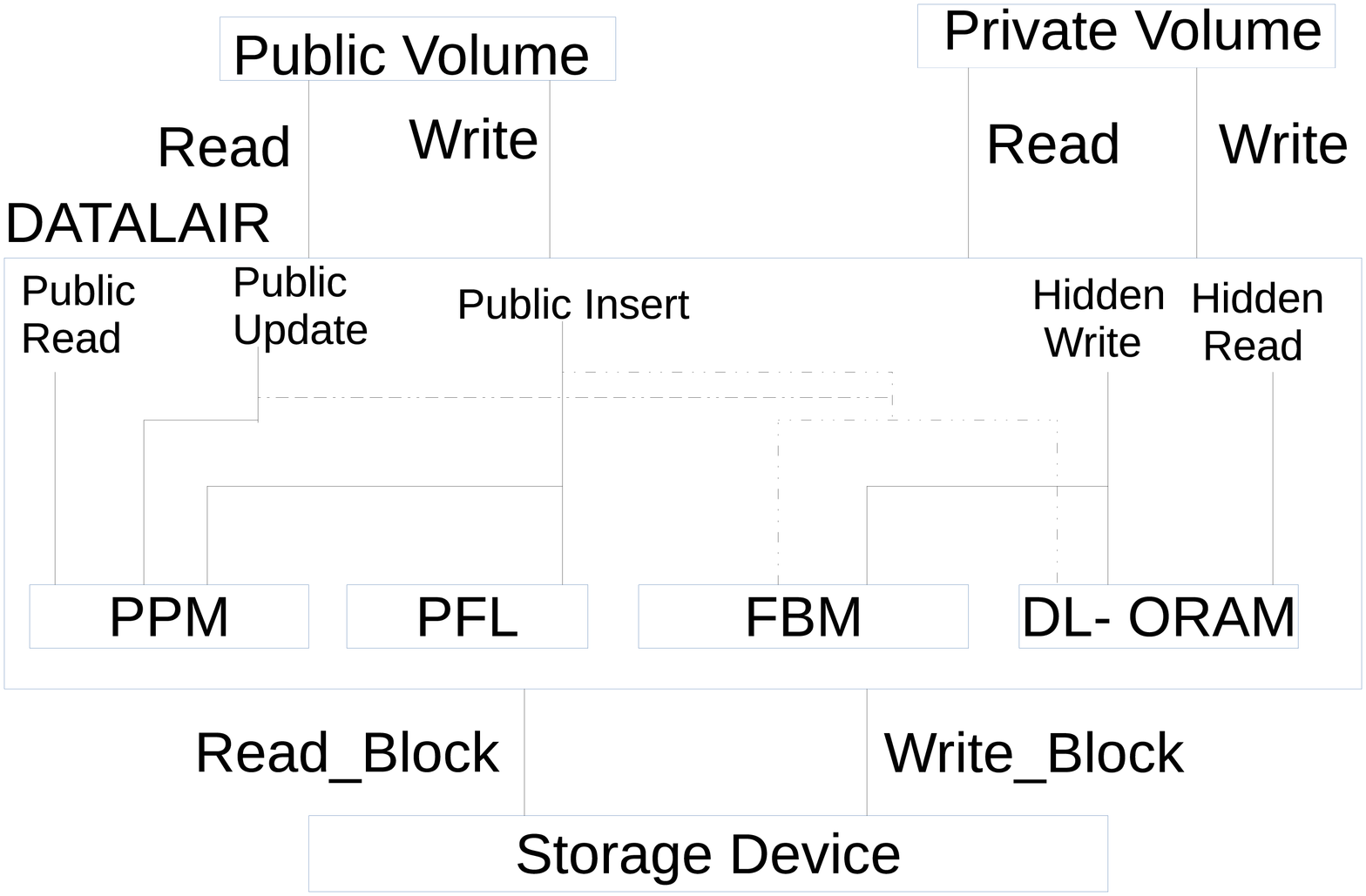}
 \vspace{-0.5cm}
 \caption{DataLair full design with the four main components: DL-ORAM, PPM, PFL and FBM. 
 Hidden data I/O is mapped through DL-ORAM while public data I/O is mapped through the PPM. 
 The public data inserts and updates simulates an access to the DL-ORAM.
 In $ONLY\_PUB$ mode free blocks are located using the PFL.
 \label{model:scheme}}
\end{figure}

\smallskip
\noindent
{\bf Optimization: in-place updates for public data.~} When a public block
is written for the first time (insert), it requires searching for a random
free block as described above, but subsequent updates can be made in-place,
thus avoiding additional accesses for finding free blocks and updating the
PPM.

\smallskip
\noindent
{\bf Storing the stash.~} The in-memory stash is stored to the disk at a graceful
power-down.  \sysname~ Full allocates a fixed location to store the
constant-sized stash.  On power-down, the stash is written encrypted to that
location.  If the stash is empty or not being used (in case of $ONLY\_PUB$
mode), \sysname~ Full writes random data instead of the stash for
indistinguishability.  On boot-up, the stash is read to memory and
reencrypted. 




Figure~\ref{model:scheme} illustrates the \sysname~ Full design. 

\smallskip
\noindent
{\bf Security.~} \sysname~ Full derives its security properties
straightforwardly from \sysname~ Lite.  First, note that the only difference
between the schemes is that public data and hidden data coexist in the same
physical address space.  Public data is mapped randomly, {\em independent}
of the locations where hidden data is already stored.  This is ensured by
the free block selection mechanism. The PPM and the PFL (added in
this construction) are public data and do not need to be protected.

Similar to the Lite construction, hidden writes through \sysoram~ are either 
performed with public writes (inserts and updates) or {\em simulated}
indistinguishably. This provides HWA indistinguishability. Further, as shown above, 
the $ONLY\_PUB$ mode of operation is indistinguishable from the $PUB\_HID$ mode 
for \sysname~ Full. This ensures PAT indistinguishability. As shown before, these 
are the necessary and sufficient conditions for PD.

\section{Evaluation}
\label{implementation}
\smallskip
\noindent
{\bf Implementation.~} 
We implemented \sysname~ Full as a kernel module device
mapper, a Linux based framework for mapping blocks in logical volumes to
physical blocks.  The default cipher used is AES-256 in counter mode with
individual per-block random IVs.  Underlying hardware blocks are 512 bytes
each and 8 adjacent hardware blocks constitute a \sysname~ ``physical
block''. Logical and physical block sizes are 4KB.

\sysname~ was benchmarked with two logical volumes (one public and one
hidden) using an ext4 filesystem in ordered mode (metadata journaling) on top.
Each volume was allocated a logical size of 25\% of the
underlying device capacity. This ensures that the combined size 
of the logical volumes is 50\% of the device, thus ensuring that 
50\% of the device is always free.
Throughput was compared against Hive \cite{Blass:2014:TRH:2660267.2660313} as 
well dm-crypt, a commonly used linux device mapper for full disk encryption. 

\smallskip
\noindent
{\bf Platform.~} Benchmarks were conducted on Linux boxes with Intel Core
i7-3520M processors running at 2.90GHz and 6GB+ of DDR3 DRAM.  The storage
device was a 1TB Samsung-850 Evo SSD. Logical volume sizes were set 64GB 
while \sysname~ was built on a 256GB physical partition. Benchmarking was 
performed using Bonnie++ \cite{bonnie++} on Ubuntu 14.04 LTS, kernel version 3.13.6. 
Benchmarking for Hive \cite{Blass:2014:TRH:2660267.2660313} was performed
with the same parameters by compiling the open source project \cite{hive}.

System caches were flushed in between tests. All tests were run 3 times 
and results collected with 95\% confidence interval.
\sysname~ stores all internal data structures persistently. 
The in-memory stash used was 
constrained to 50 4KB blocks (200KB of data).

%

\begin{table}
\begin{center}
    \begin{small}
    \begin{tabular}{ | l | l | l | l |}
    \hline
    Access & dm-crypt & \sysname~ & HIVE \cite{Blass:2014:TRH:2660267.2660313}\\ 
\hline
      Public Read  & 225.56 & 84.1 & 0.88 \\
      Public Write & 210.10 & 2.00 & 0.57 \\
      Hidden Read &  n/a & 6.00 & 5.36 \\
      Hidden Write &  n/a & 2.92 & 0.60 \\
    \hline
    \end{tabular}
    \caption{Throughput Comparison (MB/s). Higher is better. \sysname~ 
performance for public data reads is
    practical when compared to dm-crypt and almost 100x faster than existing 
work
    \cite{Blass:2014:TRH:2660267.2660313}. For hidden data writes, \sysname~ is 
5x
    faster. \label{throughput}}
   \vspace{-1.0cm}
    \end{small}
    \end{center}
\end{table}

\smallskip
\noindent
{\bf Throughput.~}
Table ~\ref{throughput} shows throughput comparison for \sysname~, HIVE and
dm-crypt. Public reads feature a throughput of about 85MB/s, 100x faster
than existing work \cite{Blass:2014:TRH:2660267.2660313} and only 2.5x slower 
than dm-crypt. The speedup results from the fact that public reads do not need to use 
the ORAM. Note that the PPM stil needs to be accessed first for determining the physical 
location of the logical block. This additional synchronous access results in the overhead 
when compared to dm-crypt.

Public writes simulate a \sysoram~ access. The improved write complexity 
of \sysoram~ compared to HIVE-ORAM~\cite{Blass:2014:TRH:2660267.2660313} 
results in a 4x speedup. Later we show how to optimize this further 
for more practical use. Similarly, 
hidden writes for \sysname~ are almost 5x faster than HIVE. 
Hidden reads perform comparably to HIVE since the overall read complexity 
is asymptotically the same for \sysoram~ and 
HIVE-ORAM~\cite{Blass:2014:TRH:2660267.2660313}. 


\begin{table}
\begin{center}
    \begin{small}
    \begin{tabular}{| l | l | l | l |}
    \hline
    Access & dm-crypt & \sysname~ & HIVE \cite{Blass:2014:TRH:2660267.2660313}\\ 
\hline
      Public Read & .007 & .018 & 1 \\
      Hidden Read  & n/a & .10 & .19 \\
      Public Write & .7 & 25 & 332 \\	
      Hidden Write & n/a & 92 & 219 \\
    \hline
    \end{tabular}
    \caption{Latency Comparison (in seconds). Lower is better. \sysname~ is 
100x faster than HIVE~\cite{Blass:2014:TRH:2660267.2660313} 
    for public reads and almost 10x faster public writes.  \label{latency}}
    \vspace{-1.0cm}
    \end{small}
    \end{center}
\end{table}

\smallskip
\noindent
{\bf Latency.~}
Table~\ref{latency} 
shows the latency comparison for \sysname~, HIVE and dm-crypt. Expectedly, 
\sysname~ public reads are almost 100x 
faster than HIVE~\cite{Blass:2014:TRH:2660267.2660313}. It also interesting to 
note that \sysname~ public writes are 
almost 15x faster than HIVE~\cite{Blass:2014:TRH:2660267.2660313}. This is due 
to the reduced number of I/Os that needs 
to be performed per access due to the better write-complexity of \sysoram~.


\smallskip
\noindent
{\bf Writing hidden data with public updates.~}
%
A straightforward 
optimization for \sysname~ is to use only public updates (in-place) for hidden 
writes/simulations, avoiding the expensive free block selection. In fact, 
since filesystem block access patterns typically follows a {\em zipfian 
distribution}~\cite{zipf_fs} -- only a small group of existing blocks are 
accessed/updated frequently --
an update to a group of already existing public data blocks is more likely than inserting new 
data. Also, since filesystems do not indicate deletes to the device, once the public volume 
is completely occupied, all subsequent operations during the lifetime of the 
disk will be treated as updates by the device mapper.



\begin{figure}
\centering
 \includegraphics[scale=0.50]{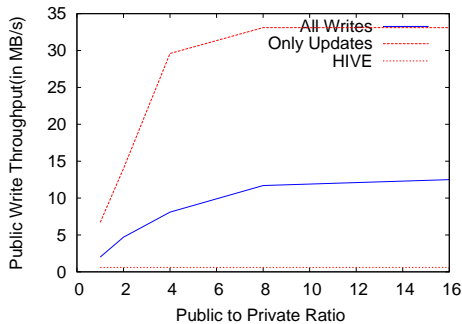}
 \vspace{-0.2cm}
 \caption{Variations in public write throughput vs. the public write 
to private write ratio. The x-axis represents the 
 number of public writes that are performed in between two private 
writes/simulation. The throughput plateaus at around 12 MB/s, around 
 6x faster than the configuration where hidden writes are made with each public 
write. The benefits of performing hidden writes only with updates is visible even in this case.
\label{phi_throughput}}
\end{figure}

\smallskip
\noindent
{\bf Frequency of Hidden Writes.~} 
%
\sysname~ features a solution for PD-CPA (Section \ref{pd}) where the number of hidden operations 
performed with each public operation, $\phi = 1$. 
For real world applications, it is reasonable to assume that a 
user will access hidden data less often that public data. In that case, $\phi$ 
can be configured according to an estimated workload to improve the public write 
throughput. 

Figure~\ref{phi_throughput} shows the variations in the public write throughput 
while increasing the public write to private write ratio. The throughput 
achieves a maximum of around 12MB/s when hidden writes/simulations 
are made every 10 public writes (inserts and updates). When hidden 
writes/simulations are performed only with updates (as described above), the 
write 
throughput achieves a maximum of around 30MB/s, around 40x faster than 
HIVE~\cite{Blass:2014:TRH:2660267.2660313}. Note that since HIVE~\cite{Blass:2014:TRH:2660267.2660313} uses a write-only ORAM 
for public writes, excluding hidden writes for a fraction of the 
public writes does not result in significant gains when compared to the overhead of the ORAM. 
Although, \sysname~ is still 7x slower 
than dm-crypt (Table~\ref{throughput}), 
the additional PD guarantees over full disk encryption makes this acceptable in 
practice.

\section{Conclusion}

This work shows that it is not necessary to sacrifice performance to achieve
plausible deniability (PD), even in the presence of a powerful multi-snapshot
adversary.  \sysname~ is a block device with practical performance and
PD assurances, designed around a new efficient write-only
ORAM construction.  \sysname~ public data reads are two orders of magnitude
faster than existing approaches while accesses to hidden data are 5
times faster. For more restricted settings, \sysname~ can achieve 
public data write performance almost 50x faster than existing work.

\section{Acknowledgement}
This work was supported by the National Science Foundation through
awards 1161541, 1318572, 1526102, and 152670. We would like to thank 
our shepherd, Giulia Fanti and the anonymous reviewers for their suggestions
on improving the paper. We also thank Vinay Ganeshmal jain for helping 
us implement \sysname~.

\vfill\eject
\appendix

\section{DL-ORAM Protocols}

\begin{algorithm}
 \SetKwInOut{Input}{input}
 \SetKwInOut{Return}{return}
  \Input{logical block address $id$}
 root:= //determine B+ tree root address from fixed location on disk\;
  $depth$ := $log_{\beta}(N)$\;
  $index$ := $\beta^{depth}$\;

 \While{not at leaf}{
      $root$ := // child \# $\floor{\frac{id}{index}}$\;
      // Search subtree rooted at $root$\;
      $blk$ = // Read physical block corresponding to $root$\;
      $depth$ := $depth - 1$\;
      $index$ := $index/\beta$\;
  }
      $addr$ := // entry for $B$ in $root$\;
      $blk$ := // Read block from disk with address $addr$\;
 \Return{Decrypt($blk$)}
  \caption{$read\_oram(id)$\label{Oramread}}
 \end{algorithm}
 \begin{algorithm}
 \small
 \SetKwInOut{Input}{input}
 \SetKwInOut{Return}{return}
 \Input{logical block address $id$,data $d$}

  root:= //determine B+ tree root address from fixed location on disk\;
  $depth$ := $log_{\beta}(N)$\;
  $index$ := $\beta^{depth}$\;
 	 
 \While{not at leaf}{
      $root$ := // child \# $\floor{\frac{id}{index}}$\;
      // Search subtree rooted at $root$\;	
      $blk$ = // Read physical block corresponding to $root$\;
      $depth$ := $depth - 1$\;
      $index$ := $index/\beta$\;
  }
	// Find free block for new write // \;
      $new\_blk\_id$ := // Find free block \; 
    
      disk.Write($new\_blk\_id$,$d$) \;
   		
   Map.updateMap($id$, $new\_blk\_id$) \;
   \caption{$write\_oram(\beta,d)$\label{Oramwrite}}
  
 \end{algorithm}

\begin{algorithm}
\small
 \SetKwInOut{Input}{input}
 \SetKwInOut{Return}{return}
 \Input{logical block address for map node $id$, physical block address where map node is written $new\_blk\_id$}
   $root$ := //Determine from fixed location\;
 \uIf{at $root$}{
 	// Update new root address at fixed location //
}
 \uElse{
 	$l$ := // READ leaf node for $id$\;
 	$id$ := // ID for leaf node\;
 	// Update $l$ with new mapping for $id$\;
 	$write\_oram(id,l)$\;
 	 }		 

	\caption{$Map.updateMap(id,new\_blk\_id)$\label{updatemap}}
\end{algorithm}

\end{document}